\documentclass{article}
\usepackage{fullpage}
\usepackage[utf8]{inputenc}
\usepackage[inline]{enumitem}
\usepackage{amsmath,amssymb}
\usepackage{xstring} 
\usepackage{booktabs}
\usepackage{ifthen}
\usepackage{multirow}
\usepackage{todonotes}
\usepackage{mathrsfs}
\usepackage{xspace}
\usepackage{authblk}
\usepackage{hyperref}
\usepackage{amsthm}
\usepackage[capitalise]{cleveref}

\newcommand\myshade{85}
\colorlet{mylinkcolor}{violet}
\definecolor{MidnightBlue}{rgb}{0.1, 0.1, 0.44}	
\definecolor{YellowOrange}{rgb}{1.0, 0.62, 0.0}
\colorlet{mycitecolor}{MidnightBlue}
\colorlet{myurlcolor}{YellowOrange}

\hypersetup{
  linkcolor  = mylinkcolor!\myshade!black,
  citecolor  = mycitecolor!\myshade!black,
  urlcolor   = myurlcolor!\myshade!black,
  colorlinks = true,
}

\usepackage{amsthm}
\newtheorem{classicaltheorem}{Classical Theorem}
\crefname{classicaltheorem}{Classical Theorem}{Classical Theorems}
\newtheorem{theorem}{Theorem}
\newtheorem{proposition}[theorem]{Proposition}
\newtheorem{remark}[theorem]{Remark}
\newtheorem{lemma}[theorem]{Lemma}
\newtheorem{example}[theorem]{Example}
\newtheorem{assumption}[theorem]{Assumption}
\newtheorem{definition}[theorem]{Definition}
\newtheorem{corollary}[theorem]{Corollary}
\numberwithin{theorem}{section}

\title{\bf
PPP-Completeness and Extremal Combinatorics\thanks{Part of this wok done while visiting R.B., L.F., P.H., and N.I.S. were visiting Bocconi University.}}

\author[1]{Romain Bourneuf}
\author[2,3]{Lukáš Folwarczný\thanks{Supported by the Grant Agency of the Czech Republic under the grant agreement no.~19-27871X and by the Charles University grant SVV–2020–260578.}}
\author[2]{Pavel Hubáček\thanks{Supported by the European Research Council (ERC) under the European Union’s Horizon 2020 research and innovation programme (Grant agreement No. 101019547), by the Cariplo CRYPTONOMEX grant,
by the Grant Agency of the Czech Republic under the grant agreement no. 19-27871X,
and by the Charles University project UNCE/SCI/004.}}
\author[4]{Alon Rosen\thanks{Supported by the European Research Council (ERC) under the European Union’s Horizon 2020 research and innovation programme (Grant agreement No. 101019547) and Cariplo CRYPTONOMEX grant.}}
\author[5]{Nikolaj I. Schwartzbach}
\affil[1]{ENS de Lyon}
\affil[2]{Charles University, Faculty of Mathematics and Physics}
\affil[3]{Institute of Mathematics, Czech Academy of Sciences}
\affil[4]{Bocconi University and Reichman University}
\affil[5]{Aarhus University}

\newcommand{\EKR}{Erdős-Ko-Rado\xspace}

\newcommand{\Sunflower}{\textsc{Sunflower}\xspace}
\newcommand{\Ramsey}{\textsc{Ramsey}\xspace}

\newcommand{\WeakErdosKoRado}{\textsc{weak-Erdős-Ko-Rado}\xspace}
\newcommand{\ErdosKoRado}{\textsc{Erdős-Ko-Rado}\xspace}

\newcommand{\IntersectingFamily}{\ErdosKoRado}
\newcommand{\WeakGeneralErdosKoRado}{\textsc{weak-general-Erdős-Ko-Rado$_k$}\xspace}
\newcommand{\GeneralErdosKoRado}{\textsc{general-Erdős-Ko-Rado$_k$}\xspace}

\newcommand{\WeakCayley}{\textsc{weak-Cayley}\xspace}
\newcommand{\Cayley}{\textsc{Cayley}\xspace}

\newcommand{\WeakSperner}{\textsc{weak-Sperner-Antichain}\xspace}
\newcommand{\Sperner}{\textsc{Sperner-Antichain}\xspace}

\newcommand{\Pigeon}{\textsc{Pigeon}\xspace}
\newcommand{\WeakPigeon}{\textsc{weak-Pigeon}\xspace}
\newcommand{\GeneralPigeon}{\textsc{General-Pigeon}\xspace}

\newcommand{\SwellColoring}{\textsc{Ward-Szabo}\xspace}
\newcommand{\WardSzabo}{\textsc{Ward-Szabo}\xspace}
\newcommand{\WeakMantel}{\textsc{weak-Mantel}\xspace}
\newcommand{\Mantel}{\textsc{Mantel}\xspace}
\newcommand{\WeakTuran}{\textsc{weak-Tur\'an$_r$}\xspace}
\newcommand{\Turan}{\textsc{Tur\'an$_r$}\xspace}

\newcommand{\Ecat}{E_\textsf{Catalan}}
\newcommand{\Ecatstr}{\Tilde{E}_\textsf{Catalan}}
\newcommand{\Dcat}{D_\textsf{Catalan}}

\newcommand{\Ecov}{E_\textsf{Cover}}
\newcommand{\Dcov}{D_\textsf{Cover}}

\newcommand{\Epruf}{E_\textsf{Prüfer}}
\newcommand{\Dpruf}{D_\textsf{Prüfer}}
\newcommand{\Epruft}{\Tilde{E}_\textsf{Prüfer}}
\newcommand{\Dpruft}{\Tilde{D}_\textsf{Prüfer}}

\newcommand{\PWPP}{\textsf{PWPP}\xspace}
\newcommand{\PPP}{\textsf{PPP}\xspace}
\newcommand{\PPA}{\textsf{PPA}\xspace}
\newcommand{\PPAD}{\textsf{PPAD}\xspace}
\newcommand{\PLS}{\textsf{PLS}\xspace}

\newcommand{\TFNP}{\textsf{TFNP}\xspace}

\newcommand{\ceil}[1]{\left\lceil#1\right\rceil}
\newcommand{\floor}[1]{\left\lfloor#1\right\rfloor}

\newcommand{\cat}{\mathbin\Vert}
\newcommand{\F}{\mathcal{F}}

\newcommand{\prob}[4][false]{\noindent \textsc{#2}
\begin{itemize}[label={$\triangleright$}]
    \item \textbf{Input}: #3
    \item \ifthenelse{\equal{#1}{true}}{
        \textbf{Solution}: #4
    }{
        \textbf{Solutions}: 
        \begin{enumerate}[label=\roman*)]
            #4
        \end{enumerate}
    }
\end{itemize}}
\begin{document}

\date{}
\maketitle

\begin{abstract}
Many classical theorems in combinatorics establish the emergence of substructures within sufficiently large collections of objects. Well-known examples are Ramsey's theorem on monochromatic subgraphs and the Erdős-Rado sunflower lemma. Implicit versions of the corresponding total search problems are known to be \PWPP-hard; here ``implicit” means that the collection is represented by a poly-sized circuit inducing an exponentially large number of objects.

We show that several other well-known  theorems from extremal combinatorics -- including Erdős-Ko-Rado, Sperner, and Cayley's formula – give rise to {\em complete} problems for \PWPP and \PPP. This is in contrast to the Ramsey and Erdős-Rado problems, for which establishing inclusion in \PWPP has remained elusive. Besides significantly expanding the set of problems that are complete for $\PWPP$ and $\PPP$, our work identifies some key properties of combinatorial proofs of existence that can give rise to completeness for these classes.

Our completeness results rely on efficient encodings for which finding collisions allows extracting the desired substructure. These encodings are made possible by the tightness of the bounds for the problems at hand (tighter than what is known for Ramsey's theorem and the sunflower lemma). Previous techniques for proving bounds in \TFNP invariably made use of structured algorithms. Such algorithms are not known to exist for the theorems considered in this work, as their proofs ``from the book'' are non-constructive.
\end{abstract}

\newpage
\pagestyle{plain}
\pagenumbering{roman}
\tableofcontents

\newpage
\pagenumbering{arabic}
\section{Introduction}

A well-known theorem by Ramsey gives a lower bound on the size of the largest monochromatic clique in any edge-coloring of the complete graph using two colors.

\begin{description}
\item[Ramsey \cite{ramsey}]
\emph{
Any edge-coloring of the complete graph on $n$ vertices with two colors contains a monochromatic clique of size at least $\frac12 \log n$. 
}
\end{description}
Ramsey's theorem gives rise to a natural computational search problem \Ramsey~\cite{Krajicek05,c_ramsey}: given a description of an edge-coloring, output the vertices of a monochromatic clique of size $\frac12\log n$. Since the theorem guarantees the existence of a monochromatic clique of this size, \Ramsey belongs to the complexity class \TFNP consisting of efficiently verifiable search problems to which a solution is guaranteed to exist~\cite{MegiddoP91}.

The computational complexity of \Ramsey very much depends on its representation. One the one hand, it is efficiently solvable when the graph is given \emph{explicitly}; a folklore proof of Ramsey's theorem gives an efficient algorithm to find such a subgraph -- see \cref{app:algo_ramsey}. On the other hand, the situation is less clear when the graph is represented \emph{implicitly}, e.g., via a Boolean circuit that, for any pair of vertices, outputs the corresponding color of the edge-coloring of the graph.\footnote{Given such a representation, it might be even hard to compute the degree of a node with respect to one of the two colors.} 

Another $\TFNP$ problem considered in the literature that is motivated by a result in extremal combinatorics arises from the well-known Erdős-Rado sunflower lemma.
\begin{description}
\item[Erdős-Rado \cite{sunflower}]
\emph{
Any family of $n$-sets of cardinality greater than $n^n n!$ contains an $n$-sunflower of size $n+1$, i.e., subsets $A_1, A_2, \ldots, A_{n+1} \in \F$ such that, for some $\Delta$, $A_i \cap A_j = \Delta$ for every distinct $A_i,A_j$. 
}
\end{description}

An instance of the total search problem \Sunflower~\cite{c_ramsey} can be implicitly represented, e.g., via a Boolean circuit that, given an index of a set in the family, outputs its characteristic vector.

In general, little is known of the complexity of the implicit variants of \Ramsey or $\Sunflower$ -- the proofs of the corresponding theorems are either non-constructive or result in inefficient (i.e., superpolynomial-time) algorithms.
Both problems are known to be \PWPP-hard, as shown by Krajíček~\cite{Krajicek05} and Komargodski, Naor, and Yogev~\cite{c_ramsey}.
This means that finding the desired substructure is at least as hard as finding collisions in an arbitrary poly-sized shrinking circuit and, hence, hard in the worst-case if collision-resistant hash functions exist.
However, they are not known to be complete for the class $\PWPP$ and the intriguing question of whether they give rise to a complexity class distinct from \PWPP has remained open for years.

\subsection{Our Results}

We explore new connections between classical theorems in extremal combinatorics and the complexity classes \PPP~\cite{ppad} and \PWPP~\cite{Jerabek16}, i.e., the classes of search problems with totality guaranteed by the (weak) pigeonhole principle. We show that \PPP and \PWPP can be characterized via a number of new \TFNP problems based on the following theorems.

\begin{description}
\item[\EKR \cite{EKR}.]
\emph{Any family of distinct pairwise-intersecting $k$-sets on a universe of size $m$ has size at most $\binom{m-1}{k-1}$.}
\item[Sperner \cite{sperner}.]
\emph{
The largest antichain, i.e., a family of subsets such that no member is contained in any other, on a universe with $2n$ elements is unique and consists of all subsets of size $n$.}
\item[Cayley \cite{cayley}.]
\emph{
There are exactly $n^{n-2}$ spanning trees of the complete graph on $n$ vertices.
}
\end{description}

Just as for \Ramsey and \Sunflower, the corresponding search problems are efficiently solvable when given explicit access to the family of objects and, again, their computational complexity is open when we consider implicit access to the structure, e.g.,
where the instance is given by a circuit that on input $i$ returns an encoding of the $i^\text{th}$ object in the collection.\footnote{Note that an implicit representation of the collection might not necessarily satisfy the assumptions of the underlying theorem. For instance, representing sets via characteristic vectors for \EKR does not ensure that they are actually $k$-sets or that they are distinct.
Importantly, such a violation could  allow evading the totality of the search problem.
Nevertheless, we can ensure totality by allowing locally verifiable evidence of a malformed representation as a solution, e.g., an index not corresponding to a $k$-set or two indices corresponding to the same set.} 
The totality of the problems we define follows from a common principle -- the instances are given via an implicit representation of a sufficiently large collection of objects (e.g., subsets for Erd\H{o}s-Ko-Rado) such that, by the corresponding theorem, there exists a small subset of these objects satisfying some efficiently verifiable property (e.g, a pair of disjoint subsets for Erd\H{o}s-Ko-Rado). 

In addition to the above completeness results, we define \TFNP problems arising from the following two results in extremal combinatorics.

\begin{description}
\item[Mantel \cite{mantel}.]
\emph{
Any triangle-free graph on $n$ vertices has at most $n^2/4$ edges.
}
\item[Ward-Szabo \cite{swell-coloring}.]
\emph{
Any edge-coloring of the complete graph on $n$ vertices with $2\le r\le\sqrt{n}$ colors must contain a bichromatic triangle.
}
\end{description}

We show that variants of the corresponding problems are hard for \PWPP and \PPP.
However, proving their inclusion in \PWPP or \PPP remains open and they join \Ramsey and \Sunflower as candidate problems that might define a new class above \PWPP or \PPP (see~\Cref{sec:OpenProblems}).
An overview of our results in terms of weak and strong problems (see Section~\ref{PPP and extremal}) is given in~\cref{table:problems}.

\begin{table}[h!]
\centering
\begin{tabular}{l|c|c}
\toprule\hline
\bf Problem             & \bf Hardness & \bf Containment \\ \midrule
$\Ramsey$            & \multirow{ 2}{*}{\parbox{1.8in}{\centering
\PWPP
\cite{Krajicek05, c_ramsey}}
} &  \multirow{ 3}{*}{\TFNP} \\ \cline{1-1}
\Sunflower          & &  \\ \cline{1-2}
\hyperref[def:problem_WS1]{$\SwellColoring$}${} $          & \multirow{4}{*}{\parbox{1.8in}{\centering $\qquad\PWPP$
\newline
\newline
[\cref{thm:swell_hard,thm:weakmantel_hard,thm:weakturan_hard}]}} \\ \cline{1-1}\cline{3-3}
\hyperref[def:problem_WeakMantel]{\WeakMantel}${} $             &  & \multirow{3}{*}{\parbox{1.8in}{\centering
$\quad\qquad\PPP$
\newline
\newline
[\cref{thm:swell2_in_ppp,thm:weakmantel_ppp,thm:weakturan_ppp}]}
} \\ \cline{1-1}
\hyperref[def:problem_WeakTuran]{$\WeakTuran$}${} $             &  & \\ \cline{1-1}
\hyperref[def:problem_WS2]{\textsc{Ward-Szabo-Colorful-Collisions}}${} $             &  &  \\ \cline{1-3}
\hyperref[def:problem_WS3]{\textsc{Ward-Szabo-Collisions}}${} $             & \multicolumn{2}{c}{\multirow{5}{*}{\PWPP\, [\cref{thm:swell3_in_pwpp,thm:swell_hard,thm:weaksperner_complete,thm:cayley_complete,thm:if_complete,thm:weakgeneralekr}]}}  \\ \cline{1-1}
\hyperref[def:problem_WeakEKR]{\WeakErdosKoRado}${} $            \\ \cline{1-1}
\hyperref[def:problem_WeakGeneralEKR]{\WeakGeneralErdosKoRado}${}$   \\  \cline{1-1}
\hyperref[def:problem_WeakSperner]{\WeakSperner}${} $  \\ \cline{1-1}
\hyperref[def:problem_WeakCayley]{\WeakCayley}${} $    \\\hline
\hyperref[def:problem_EKR]{\ErdosKoRado}${} $   & \multicolumn{2}{c}{\multirow{4}{*}{\PPP\, [\cref{thm:cayley'_complete,thm:sperner'_complete,thm:if'_complete,thm:generalekr}]}}  \\ \cline{1-1}
\hyperref[def:problem_GeneralEKR]{\GeneralErdosKoRado}${}$   \\ \cline{1-1}
\hyperref[def:problem_Sperner]{\Sperner}${} $ \\\cline{1-1}
\hyperref[def:problem_Cayley]{\Cayley}${} $\ \\\hline 
\hyperref[def:problem_Mantel]{\Mantel}${} $\ & {\PPP \,[\cref{thm:mantel_hard}]} & \TFNP \\\hline
\bottomrule
\end{tabular} 
\caption{Summary of the complexity of problems we consider.
Except for \Ramsey and \Sunflower, all problems were introduced in this work.
The containment results for $\WeakGeneralErdosKoRado$ and $\GeneralErdosKoRado$ rely on the efficient Baranyai assumption (Assumption \ref{assump:baranyai}).}
\label{table:problems}
\end{table}

\subsection{Techniques and Ideas} 
A long-standing open problem regarding \Ramsey and \Sunflower has been to determine their status with respect to the classes \PWPP and \PPP.
%
%
For the most part, the most challenging part in establishing completeness for some syntactic subclass of \TFNP lies in proving hardness (see, e.g., \cite{DaskalakisGP09,Mehta18,Filos-RatsikasG18}).
For subclasses of \TFNP such as \PPAD, \PPA, and \PLS, the inclusion in a subclass mostly follows from the existence of an inefficient yet structured algorithm for the problem at hand; for example, the chessplayer algorithm for \PPA~\cite{ppad} or the steepest descent algorithm for \PLS~\cite{JohnsonPY88}.
However, this methodology seems inapplicable for proving inclusion in \PWPP or \PPP as these classes do not exhibit any characterizing graph-theoretic structure that could capture some class of natural algorithms. 


In contrast to many existing bounds in \TFNP, our work does not make use of structured algorithms but instead makes use of encodings that translate between substructures and collisions in circuits. In order to establish inclusion in $\PWPP$, we encode the objects of the collection using a ``property-preserving encoding" that encodes the objects in a way that translates some specific relation into collisions.
More precisely, we want an encoding function that is efficiently computable and (nearly) optimal, such that whenever two elements have the same encoding, these two elements give a solution to the original problem.
While this technique is quite general, it is not always clear how to instantiate the encoding to get the desired collisions.

Consider, for example, the total search problem corresponding to the \EKR theorem for intersecting families of $n$-sets on a universe of size $2n$.
An instance can be given by a Boolean circuit
$C\colon\{0,1\}^{\ceil{\log\left(\binom{2n-1}{n-1}\right)} + 1} \to \{0,1\}^{2n}$ representing a family of subsets of $[2n]$, i.e., $C(i)$ is the characteristic vector of the $i$-th $n$-set in the family.
Suppose the outputs of $C$ define  distinct $n$-sets.
Since there are more than $\binom{2n-1}{n-1}$ of them, then, by the \EKR theorem, there must exist a pair of inputs mapped to disjoint $n$-sets by $C$.
We define any such pair of inputs to be a solution.\footnote{To ensure the totality of the problem, we introduce additional solutions corresponding to succinct certificates that $C$ does not define a family of distinct $n$-sets, i.e., either an $i$ such that $C(i)$ is not of Hamming weight $n$ or a pair $i\neq j$ such that $C(i)=C(j)$.}

When proving that the above total search problem is contained in the complexity class \PWPP, at a high level, we want to encode the $n$-sets of the family using a shrinking circuit, in such a way that collisions correspond to disjoint sets.
Observe that for $n$-sets in a universe of size $2n$, the only disjoint sets are complements and, hence, we get an equivalent instance of the problem if we map each set to either itself or its complement, arbitrarily.
In our construction, we map each set $S$ to the representative not containing the element 1.
That is, if $1\not\in S$, the set is left unchanged and, otherwise, it is mapped to its complement $\overline{S}$.
Note that, by the pigeonhole principle, two sets that do not contain 1 must have a non-empty intersection since we work with $n$-subsets of $[2n]$.
To obtain a shrinking circuit, we make use of Cover encodings (\cref{sec:cover}) that give an optimal encoding of all $n$-sets by considering their lexicographic order.
Notice that if the input $S$ is not an $n$-set, we may map it arbitrarily to any $n$-set, as a collision, in this case, yields a solution to the instance of the above problem motivated by the \EKR theorem. 

In contrast, the \PWPP-hardness results for \Ramsey and \Sunflower follow an extremely elegant but rather direct (compared to other hardness results for subclasses of \TFNP) technique of graph-hash product~\cite{Krajicek05,c_ramsey},
which we illustrate on \Ramsey.
Recall that there are known randomized constructions of edge-colorings of the complete graph $K_{2^{n/4}}$ on $2^{n/4}$ vertices that do not contain a monochromatic clique of size $n/2$~\cite{erdos1947ramsey}.
Given such an underlying edge-coloring of $K_{2^{n/4}}$ and a hash function $h$ mapping $n$-bit strings to $n/4$-bit strings, one can construct an edge-coloring of the complete graph on $2^n$ vertices by assigning to every edge $(u,v)\in\{0,1\}^n\times\{0,1\}^n$ the color of the edge $(h(u),h(v))\in\{0,1\}^{n/4}\times\{0,1\}^{n/4}$ from the underlying coloring.
Since the underlying edge-coloring of $K_{2^{n/4}}$ does not contain a monochromatic clique of size $n/2$, it is easy to see that any monochromatic clique of size $n/2$ in the resulting edge-coloring of $K_{2^n}$ (guaranteed to exist by Ramsey's theorem) must have been introduced via a collision in the hash $h$.

As noted by~\cite{c_ramsey}, the structure of a \PWPP-hardness proof using the graph-hash product is not restricted to total search problems corresponding to graph-theoretic theorems of existence;
indeed, \cite{c_ramsey} used the graph-hash product to prove also \PWPP-hardness of \Sunflower.
On a high level, for a problem to be amenable to the graph-hash product technique, it is sufficient to be able to construct a collection of objects such that 1) it does not contain the desired substructure, 2) its size is at least a constant fraction of the threshold necessary for the existential theorem to apply,\footnote{This is a technical condition ensuring that we can reduce from a \PWPP-complete variant of the problem of finding collisions in a shrinking hash. Note that it is easy to find collisions in functions that exhibit extreme shrinking.} and 3) it can be efficiently indexed.
Then, we can interpret the output of an appropriately shrinking hash $h$ as an index into the small collection of objects, and, for each index, we can efficiently compute and output the corresponding element in the collection.
Again, since the small collection does not contain the desired substructure, all solutions of the instance constructed via graph-hash product must in some way result from a collision in the hash $h$.

For example, consider the total search problem arising from Sperner's theorem on antichains --
here, the threshold size is $\binom{2n}{n}$, meaning that if we have a family with strictly more than $\binom{2n}{n}$ distinct subsets of $[2n]$ then one subset from the family must be contained in another member of the family.
It is straightforward to construct a family of subsets that does not contain the specific substructure (i.e., a subset that is included in another one) with size equal to the threshold size $\binom{2n}{n}$.
It suffices to consider the family of all the $n$-subsets of $[2n]$.
Similarly, for many other combinatorial problems we study, an adequate collection of objects can be found by looking at a collection of maximum size that does not contain the substructure.

We also show natural reductions between some of the problems we define (from $\ErdosKoRado$ to $\Sperner$ for instance), which, in our opinion, highlights the relevance of these new problems and the fact that their definition is the correct one.

\subsection{\PPP-Completeness From Extremal Combinatorics}
\label{PPP and extremal}
Up to this point, our discussion did not explicitly distinguish between the classes \PWPP and \PPP.
However, our work highlights important structural differences between the two complexity classes.
Recall that the class \PWPP contains the search problems in \TFNP whose totality can be proved using the \emph{weak} pigeonhole principle: ``In any assignment of $2n$ pigeons to $n$ holes there must be two pigeons sharing the same hole."

This statement can be seen as a result in extremal combinatorics bounding the maximum number of pigeons that can be assigned to $n$ holes without two pigeons being sent to the same hole.
More generally, we say that a theorem from extremal combinatorics is ``weak" if it gives an upper bound (which may or may not be tight) on the maximum size of a collection of objects that does not contain some substructure (above, two pigeons sharing the same hole).
On the contrary, we say that a theorem from extremal combinatorics is ``strong" if it gives a tight upper bound on the maximum size of a collection of objects that does not contain some substructure, as well as some structural property about the maximum families without the substructure.
For instance, the strong pigeonhole principle can be stated as: ``In any assignment of $n$ pigeons to $n$ holes there is either a pigeon in the first hole or two pigeons sharing the same hole."
Note that it is exactly this formulation of the strong pigeonhole principle that defines the class \PPP.
    
Many results in extremal combinatorics have a weak statement and a strong statement. For such results, we can define a problem corresponding to the weak statement, which often is related to $\PWPP$, and a problem corresponding to the strong statement, which often is related to $\PPP$. In this paper, all $\PWPP$-hard problems correspond to a weak theorem in extremal combinatorics, while $\PPP$-hard problems correspond to a strong theorems in extremal combinatorics. As an example, consider Cayley's formula and note that the bound $n^{n-2}$ is tight. Hence, if we are given a collection of exactly $n^{n-2}$ distinct graphs on $n$ vertices, then either one of the graphs is not a spanning tree, or every spanning tree is in the collection. This observation induces a \TFNP problem that we show to be \PPP-complete. 

\subsection{Related Work}
\label{sec:RelatedWork}

Compared to the majority of subclasses of \TFNP that have been extensively studied and are known to capture various total search problems from diverse domains of mathematics, \PPP and \PWPP might seem less expressive and the first non-trivial completeness results appeared only recently.

Sotiraki, Zampetakis, and Zirdelis~\cite{SotirakiZZ18} and 
Ban, Jain, Papadimitriou, Psomas, and Rubinstein~\cite{BanJPPR19} demonstrated that \PPP contains computational problems from number theory and the theory of integral lattices.
In particular, Sotiraki et al.  showed \PPP-completeness of a computational problem related to Blitchfeld's theorem and \PPP-completeness (resp. \PWPP-completeness) of a problem motivated by the Short Integer Solution problem.
Hubáček and Václavek~\cite{HubacekV21}
showed that some general formalizations of the discrete logarithm problem are complete for \PWPP and \PPP and, motivated by classical constructions of collision-resistant hashing, they characterized \PWPP via the problem of breaking claw-free (pseudo-)permutations.


\subsection{Open Problems}
\label{sec:OpenProblems}
Our work suggests various interesting directions for future research:

\begin{itemize}
\item
We exploit the power of strong statements in extremal combinatorics for establishing \PPP-completeness.
The notorious lack of tight bounds for the Erd\H{o}s-Rado sunflower lemma and Ramsey's theorem implies that we have no strong version of these theorems, which may explain why showing the inclusion of the corresponding problems in, e.g., \PPP has eluded researchers.

\item
We introduced total search problems corresponding to Mantel's theorem, Tur\'an's theorem, and Ward-Szabo's theorem.
In this work, we only prove hardness results for these problems but no inclusion results.
Hence, it is still open whether they are complete for the classes $\PPP$ and $\PWPP$, or whether they could define a new subclass of \TFNP.

\item
The $\Turan$ problem is defined in a similar fashion to $\Mantel$, yet, unlike for $\Mantel$, we currently do not have a proof of \PPP-hardness for it.
Thus, the question of \PPP-hardness of $\Turan$ is immediate.
Alternatively, it would be interesting to define a different $\PPP$-hard problem in a natural way from Tur\'an's theorem.

\item
Another exciting question is whether the efficient Baranyai assumption (\Cref{assump:baranyai}) holds, as well as whether it is possible to prove the inclusion results of the problems associated to the general version of Erd\H{o}s-Ko-Rado's theorem \emph{without} that assumption.
Showing reductions between \textsc{general-Erdős-Ko-Rado$_k$} and \textsc{general-Erdős-Ko-Rado$_l$} for $k \neq l$ without the efficient Baranyai assumption would also be intriguing.

\item
Finally, we believe the problems $\GeneralPigeon_k^m$ deserve a more thorough investigation to further our understanding of the classes they define and their interrelation.

\end{itemize}

\section{Preliminaries}
\label{sec:Preliminaries}
We denote by $\log x$ the binary logarithm of $x$.
We denote by $[n]$ the set $\{1, 2, 3, \ldots, n-1, n\}$.
We interpret elements of $\{0,1\}^*$ as strings and write them as $x=x_1x_2\cdots x_n$ for $x_i \in \{0,1\}$.
Each element $x_i$ is also called a bit.
We say $n$ is the length of $x\in\{0,1\}^n$, and say $x$ is an $n$-bit string.
We denote by $0^n$ (resp. $1^n$) the $n$-bit string consisting of all 0 (resp. 1).
If $x,y \in \{0,1\}^*$ are two strings of lengths $n,m$, respectively, we denote by $x \cat y=x_1x_2\cdots x_ny_1y_2\cdots y_m$ the \emph{concatenation} of $x$ and $y$. 
We denote by $\leq$ the lexicographical order on strings.
Note that $\leq$ is a partial order as it is only well-defined for strings of the same length.
We use $x<y$ to denote $x\leq y$ and $x\neq y$.
We may occasionally abuse notation and write $x<k$ where $k\in\mathbb{N}$, in which case we mean the binary encoding of $k$ on the same number of bits as $x$.
If $\ceil{\log k}$ exceeds the length of $x$, we define $x<k$ such that the order is total.

If $\Omega$ is a set of size $n$, we associate the set $2^\Omega$ with the characteristic vectors from $\{0,1\}^n$ for some arbitrary (but fixed) order on $\Omega$.
We denote by $\subseteq$ the partial order on $\{0,1\}^n$ where $x \subseteq y$ iff $x_i \leq y_i$ for every $i=1 \ldots n$. If $x \in \{0,1\}^n$ is a string, we denote by $\overline{x}:=\overline{x}_1\overline{x}_2\cdots \overline{x}_n$ the \emph{complement of $x$}, defined by $\overline{x}_i = 1 - x_i$.
We also use other set-theoretic operators $\cap, \cup, \setminus$ that are defined in a natural way.
We also denote by $|x|=\sum_{i=1}^n x_i$ the number of 1s in $x$ when the length is implicit from the context.

\subsection{Total Search Problems}
A search problem is defined by a binary relation $R \subseteq \{0,1\}^* \times \{0,1\}^*$ --
a string $s\in\{0,1\}^*$ is a \emph{solution} for an \emph{instance} $x\in\{0,1\}^*$ if $(x,s) \in R$.
A search problem defined by relation $R$ is \emph{total} if for every $x$, there exists an $s$ such that $(x,s)\in R$.
We define $\TFNP$ as the class of all total search problems that can be efficiently verified, i.e., there is a deterministic polynomial-time Turing machine that, given $(x,s)$, outputs 1 if and only if $(x,s) \in R$ and, for every instance $x$, there exists a  solution $s$ of polynomial length in the size of $x$.

To avoid unnecessarily cumbersome phrasing throughout the paper, we define \TFNP relations implicitly by presenting the set of valid \emph{instances} $X\subseteq \{0,1\}^*$ recognizable in polynomial time (in the length of an instance) and, for each instance $i\in X$, the set of admissible \emph{solutions} $Y_i\subseteq\{0,1\}^*$ for the instance $i$.
It is then implicitly assumed that, for any invalid instance $i\in\{0,1\}^*\setminus X$, we define the corresponding solution set as $Y_i=\{0,1\}^*$.

Next, we recall the definitions of the complexity classes \PWPP and \PPP via their canonical complete problems \WeakPigeon and \Pigeon.

\begin{definition}[\WeakPigeon and \PWPP~\cite{Jerabek16}]\label{def:problem_WeakPigeon}

The problem \WeakPigeon is defined by the relation
\begin{description}
    \item[Instance:] A Boolean circuit $C\colon\{0,1\}^n \to \{0,1\}^{n-1}$.
    \item[Solution:] $x_1 \neq x_2$ s.t. $C(x_1) = C(x_2)$.
\end{description}
The class of all \TFNP problems reducible to \WeakPigeon is called \PWPP.
\end{definition}
    
\begin{definition}[\Pigeon and \PPP~\cite{ppad}]\label{def:problem_Pigeon}
The problem $\Pigeon$ is defined by the relation
\begin{description}
    \item[Instance:] A Boolean circuit $C\colon\{0,1\}^n \to \{0,1\}^{n}$.
    \item[Solution:] One of the following:
    \begin{enumerate}[label=\roman*)]
    \item $x$ s.t. $C(x) = 0^n$,
    \item $x \neq y$ s.t. $C(x) = C(y)$. 
    \end{enumerate}
\end{description}
\medskip
The class of all \TFNP problems reducible to \Pigeon is called \PPP.
\end{definition}

\section{Property-Preserving Encodings}

A key ingredient to our proofs of inclusion in \PWPP and \PPP is the use of efficient encodings.
We rely on two different types of encodings.
The first one simply consists of bijections between two different representations of the same set of objects, the first one being more natural and more convenient to work with, and the second one being more concise.
The second type of encodings, which we call property-preserving encodings, consists of shrinking functions, in the sense that the range of the encoding is smaller than the domain, whose collisions exactly correspond to elements sharing some property.
The following definition gives a precise description of the features we require from these encodings.

\begin{definition}[Property-preserving encoding]
    Let $\mathcal{X} \subseteq \{0,1\}^k,\mathcal{Y}$ be sets, and let $\sim$ be an equivalence relation on $\mathcal{X}$. Let $E : \{0,1\}^k \rightarrow \mathcal{Y}$ be a surjection. We say that $E$ constitutes a \emph{property-preserving encoding for $\sim$ on $\mathcal{X}$} if it satisfies.
    \begin{itemize}
        \item \emph{(Efficiency)}. $E$ can be computed in polynomial time.
        \item \emph{(Compression)}. $|\mathcal{Y}| \leq |\mathcal{X}|$.
        \item \emph{($\sim$-correctness)}. $E$ is constant on every coset of $\mathcal{X}$ for $\sim$.
    \end{itemize}
\end{definition}

We first describe some bijective encodings before studying some property-preserving encodings.

\subsection{Cover Encodings}\label{sec:cover}
Our reductions in~\Cref{sec:ErdosKoRado} make use of \emph{Cover encodings} \cite{Cover} that efficiently encode subsets of a specified size in optimal space: namely, we may encode every subset $S \subseteq \{0,1\}^m$ such that $|S| = k$ by considering the lexicographic order of all $\binom{m}{k}$ such sets (in fact we consider the lexicographic order over their characteristic vectors $\in \{0, 1\}^{m}$), and mapping this into binary strings: this requires $\ceil{\log \binom{m}{k}}$ bits, which is optimal. We denote the encoding and decoding functions as follows, with $\alpha(k, m) = \ceil{\log \binom{m}{k}}$.
\begin{align*}
    \Ecov^{k,m} : \{0,1\}^{m} &\rightarrow \{0,1\}^{\alpha(k, m)}\\
    \Dcov^{k,m} : \{0,1\}^{\alpha(k, m)} &\rightarrow \{0,1\}^{m}
\end{align*}

We set $\Ecov = \Ecov^{n, 2n}$ and $\Dcov = \Dcov^{n, 2n}$, and $\alpha = \alpha(n, 2n)$. As described in $\cite{Cover}$, these functions can be made efficient.
\begin{lemma}\label{lemma:cover_bij}
For every $k \leq m$, $\Dcov^{k, m} \circ \Ecov^{k, m}$ is the identity over all $k$-subsets of $\{0,1\}^{m}$. Similarly, $\Ecov^{k, m} \circ \Dcov^{k, m}$ is the identity over the first $\binom{k}{m}$ elements in the lexicographic order of $\{0, 1\}^{\alpha(k, m)}$.
\end{lemma}
\noindent Note that the behavior of $\Dcov^{k, m}$ is undefined for the last $2^{\alpha(k, m)} - \binom{m}{k}$ inputs. Furthermore, by design, $\Ecov^{k, m}$ is well-defined on any subset of $[m]$ (even if this subset does not have size $k$), but the encoding only makes sense for subsets of size $k$. We also note the following identity which will be useful later when dealing with $n$-subsets of $[2n]$.
\begin{equation}\label{eq:decode_0_string}
    \Dcov(0^\alpha) = 0^{n} 1^n = \overline{[n]}
\end{equation}
\begin{remark} \label{rmk:cover_0s}
When we encode $n$-subsets of $[2n]$, since we encode sets according to their rank of their characteristic vector in the lexicographic order, any set that does not contain element 1 is one of the $\binom{2n-1}{n-1} = \frac{1}{2}\binom{2n}{n} \leq 2^{\alpha - 1}$ first ones in the lexicographic order, hence its encoding starts with a 0. Conversely, if we decode an element whose first two bits are 0's, this means that the corresponding $n$-subset of $[2n]$ is one of the first $2^{\alpha - 2} \leq \binom{2n-1}{n-1}$ in the lexicographic order, hence that it does not contain the element 1. $\hfill\diamond$
\end{remark}

\subsection{Encoding 2-subsets of $[2^n]$} \label{sec:elex}
In~\Cref{sec:WardSzabo}, we need to encode the subsets of $[2^n]$ with 2 distinct elements in an injective way. 
Unfortunately, since the base set is large, we cannot use Cover encodings to do so.
However, we can use the idea behind Cover encodings, that is to encode the subsets by their rank in the lexicographic order. 
Consider $(x, y) \in [2^n] \times [2^n]$, with $x < y$. What is its rank in the lexicographic order? \\
All subsets whose smallest element is smaller than $x$ have a lower rank. The number of such subsets is \begin{align*}
    (2^n-1) + (2^n - 2) + \ldots + (2^n-x+1) &= \sum_{j = 2^n-x+1}^{2^n-1} j \\
&= \sum_{j = 1}^{2^n-1} j - \sum_{j = 1}^{2^n-x} j \\
&= \frac{2^n(2^n-1)}{2} - \frac{(2^n-x)(2^n-x+1)}{2}
\end{align*}
All subsets whose smallest element is $x$ and whose second smallest element is smaller than $y$ also have a lower rank. There are exactly $y-x-1$ such subsets. \\
Hence, the rank of the subset $(x, y)$ in the lexicographic order is $$\frac{2^n(2^n-1)}{2} - \frac{(2^n-x)(2^n-x+1)}{2} + y-x-1$$ \\
Note that since there are $\binom{2^n}{2} < 2^{2n-1}$ subsets of $[2^n]$ with 2 distinct elements, the rank of any subset $(x, y)$ with $x < y$ can be written in binary using $2n-1$ bits.
Now, denote as $E_{lex} : \{0, 1\}^{n} \times \{0, 1\}^n \rightarrow \{0, 1\}^{2n-1}$ the following circuit. On input $(x, y)$, it proceeds as follows. \begin{enumerate}
    \item If $x = y$, it returns $0^{2n-1}$.
    \item If $x < y$, it computes and returns the binary encoding on $2n-1$ bits of $\frac{2^n(2^n-1)}{2} - \frac{(2^n-x)(2^n-x+1)}{2} + y-x-1$.
    \item If $x > y$, it computes and returns the binary encoding on $2n-1$ bits of $\frac{2^n(2^n-1)}{2} - \frac{(2^n-y)(2^n-y+1)}{2} + x-y-1$.
\end{enumerate}
Note that $E_{lex}$ has polynomial size, and is injective on the set of subsets of $[2^n]$ with 2 distinct elements by construction.

\begin{remark}
In fact, this encoding is a bijection from the set of 2-subsets of $[2^n]$ to the set $[\binom{2^n}{2}]$. The reciprocal of that bijection can also be computed by a circuit $D_{lex}$ of polynomial size.
\end{remark}

\subsection{Prüfer Codes}\label{sec:prufer}
In~\Cref{sec:Cayley}, we make use of Prüfer codes \cite{prufer} that give an efficiently computable bijection between the set of labelled spanning trees on $n$ vertices and the set of sequences of $n-2$ elements of $[n]$. They were originally used by Heinz Prüfer \cite{prufer} to prove \cref{cthm:cayley}

We denote by $\Epruf$ a circuit that efficiently computes the Prüfer encoding of a spanning tree described by an element of $\{0, 1\}^{\binom{n}{2}}$.  Similarly, let $\Dpruf$ be a circuit that efficiently computes the spanning tree associated with a Prüfer code. By looking at the algorithm to compute Prüfer encodings, it is clear that we can assume these circuits to have polynomial size. We also assume that $\Epruf$ outputs elements of the right form even on inputs which do not correspond to spanning trees. Consider the lexicographic order on $[n]^{n-2}$. Let $R$ be a circuit that efficiently computes the rank of an element of $[n]^{n-2}$, and let $\Epruft = R \circ \Epruf$. Given a spanning tree, $\Epruft$ returns the rank of its Prüfer code in the lexicographic order.

Let $R'$ be a circuit which on input $x$ computes the sequence of $[n]^{n-2}$ whose rank in the lexicographic order is $x$. Let $\Dpruft = \Dpruf \circ R'$. Given a rank, $\Dpruft$ returns the spanning tree whose Prüfer code has the corresponding rank in the lexicographic order. Note that $\Dpruft$ and $\Epruft$ both have polynomial size. Now, if $\beta = \lceil(n-2)\log(n)\rceil$, then $\Epruft : \{0, 1\}^{\binom{n}{2}} \rightarrow \{0, 1\}^{\beta}$, $\Dpruft : \{0, 1\}^\beta \rightarrow \{0, 1\}^{\binom{n}{2}}$. By construction, we have the following.

\begin{lemma}\label{lemma:epruft_dpruft}  The following statements are true.
\begin{enumerate}
    \item $\Dpruft \circ \Epruft$ is the identity over the set of labelled spanning trees on $n$ vertices.
    \item $\Epruft \circ \Dpruft$ is the identity over the first $n^{n-2}$ elements of $\{0, 1\}^\beta$. 
\end{enumerate}
\end{lemma}

\begin{remark}
The behavior of $\Dpruft$ on its last $2^{\beta} - n^{n-2}$ inputs is undefined.
\end{remark}

\begin{remark} \label{rmk:prufer_0's}
Let $T_1$ be the tree composed of the edges $(1, 2), (1, 3), \ldots, (1, n)$. Then, $\Epruft(T_1) = 0^\beta$ and $\Dpruft(0^{\beta}) = T_1$.$\hfill\diamond$
\end{remark}

\subsection{Catalan Factorization}\label{sec:catalan}
\emph{Catalan factorization} \cite{catalan} is an encoding of subsets of $[2n]$ that allows us to decompose the partially ordered set $(2^{[2n]}, \subseteq)$ into $\binom{2n}{n}$ chains and to move efficiently within each chain to find a canonical representative, namely the only $n$-subset of the chain.

Let $x \in \{0,1\}^{2n}$ be a bitmap representing an element of $[2n]$. We introduce a new symbol $z$, and construct the Catalan factorization as follows. We temporarily record for each symbol whether or not it is underlined.
\begin{enumerate}
    \item Underline the leftmost substring that starts with a non-underlined 1, followed by a (possibly empty) sequence of underlined symbols, and ends in a non-underlined 0. If no such substring exists, go to step 3.
    \item Go to step 1.
    \item Record the number $k$ of non-underlined 1's. 
    \item Replace all non-underlined symbols in $x$ with $z$, and let $x' \in \{0,1,z\}^{2n}$ be the resulting string (with underlinings removed).
    \item Output $(x',k)$.
\end{enumerate}
We denote the output of the Catalan factorization as $\Ecat(x)\in \{0,1,z\}^{2n} \times [2n]$. We say $x'=\Ecatstr(x)$ is the \emph{Catalan string} of $x$. If $x' \in  \{0, 1, z\}^{2n}$ and $m$ is the number of $z$'s in $x'$, then for any $l \leq m$, we define $\Dcat(x', l)$ as the string obtained from $x'$ by replacing the $l$ last $z$'s by $1$ and the rest by 0. 

\begin{example}
    Let $n=4$ and let $x = 01101100$ be the string corresponding to the set $\{2,3,5,6\}$. Then, we construct the Catalan factorization by repeating step 1 to get the underlined version.
    \begin{align*}
        01101100 \rightarrow 01\underline{10}1100 \rightarrow 01\underline{10}1\underline{10}0 \rightarrow 01\underline{10}\underline{1\underline{10}0}
    \end{align*}
    We terminate as there are no non-underlined 0's with a 1 on its left. We record that there is $k=1$ non-underlined 1. We then replace all non-underlined symbols with $z$ to obtain the Catalan factorization.
    $$
        (x',k) = (zz101100, 1)
    $$
    Note that we have $\Dcat(x', k) = 01101100 = x$ so the encoding and decoding operations behave as expected. Note also that $\Dcat(x', 0) = 00101100$ corresponds to the set $\{3, 5, 6\}$ and $\Dcat(x', 2) = 11101100$ corresponds to the set $\{1, 2, 3, 5, 6\}$.
    For this reason, we say that the Catalan string $x'$ identifies the following chain. 
    $$
        \{3, 5, 6\} \subset \{2, 3, 5, 6\} \subset \{1, 2, 3,5,6\}
    $$
    In that chain, $k$ identifies that $x$ is the $1^\text{st}$ element, counting from 0. $\hfill\diamond$
\end{example}

\begin{lemma} \label{lemma:cat_bij}
$\Dcat \circ \Ecat$ acts as identity over $\{0, 1\}^{2n}$.    
\end{lemma}
\begin{proof}
Let $x \in \{0, 1\}^{2n}$, and $(x', k) = \Ecat(x)$ be its Catalan factorization. Let $m$ be the number of $z$'s in $x$. We claim that at the end of the underlining phase of the Catalan factorization of $x$, the entries that are not underlined are first $m-k$ 0's and then $k$ 1's. Indeed, by definition, $k$ of them are 1, so $m-k$ of them are 0. Furthermore, if we had a non-underlined 1 before a non-underlined 0, then we could consider the rightmost non-underlined 1 that is before a non-underlined 0. This 1 is followed by a sequence of underlined symbols and then by a non-underlined 0 so this 1 and the corresponding 0 should have been underlined. Thus, we indeed have that the entries that are not underlined are first $m-k$ 0's and then $k$ 1's. These are the entries that are turned into $z$'s when we go from $x$ to $x'$.

Now, when we compute $\Dcat(x', k)$, we replace the last $k$ $z$'s in $x'$ by 1's and the $m-k$ other ones by 0's, which is exactly what we had in $x$. Hence, $\Dcat \circ \Ecat(x) = \Dcat(x', k) = x$.
\end{proof}

We also denote by $\Dcat^{(l)} : \{0,1,z\}^{2n}\rightarrow\{0,1\}^{2n}$ the map $x' \mapsto \Dcat(x', l)$.
If on input $x'$, $l$ is larger than the number of $z$ symbols in $x'$, all $z$ symbols are be replaced with 1; this ensures the map is defined for all $l\geq 0$. 

\begin{lemma}\label{lemma:catalan_id2}
    For every $l \geq 0$, $\Ecatstr\circ\Dcat^{(l)}$ acts as identity on the set of Catalan strings. That is, if $x'$ is a Catalan string, then for every $l$, the Catalan string of $\Dcat^{(l)}(x')$ is $x'$.
\end{lemma}
\begin{proof}
    Let $x \in \{0, 1\}^{2n}$ and let $x' = \Ecatstr(x)$ be the Catalan string of $x$. Now let $l \geq 0$, $y = \Dcat(x', l)$ and $y'=\Ecatstr(y)$ be the Catalan string of $y$. We want to show that $y' = x'$. 
    
    We proceed using induction on the steps of the algorithm. At first, no entries are underlined in either string. Next, suppose that after some number of steps, the underlined bits are exactly the same in $x$ and in $y$. 
    Now, consider two bits that get underlined in $x$ at the next step. Then, all the bits between them are underlined in $x$ at this point, so this is also the case in $y$ by induction hypothesis. Furthermore, since these two bits get underlined in $x$, they are not turned into $z$'s at the end of the algorithm, which means that they are still the same bits in $x'$ and therefore in $y$. Hence, in $y$ we have these 2 bits, first a 1 and then a 0, such that every entry between them is underlined, so they get underlined at this step. 
    
    Conversely, consider two bits that get underlined in $y$ at the next step. Then, all the entries between them in $y$ are underlined at this point, so it is the case in $x$ too by induction hypothesis. By contradiction, suppose that the corresponding bits in $x$ do not get underlined at this step. By the previous observation, it means that this pair of bits in $x$ is not $(1, 0)$. There are three cases to consider: 
    \begin{enumerate}
        \item In $x$, these two bits are $0$'s. Then, the first gets turned into a 1 in $y$, which means that it never gets underlined in $x$ (otherwise it would remain the same). Then, since all the bits in $x$ between these two are already underlined, and since the first never gets underlined, this means that the second never gets underlined (there will never be a non-underlined 1 before it such that all entries between them are underlined). Hence, these two bits never get underlined in the algorithm, and are finally turned into $z$'s. Then, to go from $x'$ to $y$, we replace the $l$ last $z$'s by 1's and the others by $0$'s, thus making it impossible for the first of these two bits to be turned into a 1 while the second is turned into a 0.
        \item In $x$, these two bits are respectively 0 and 1. Then, both these bits are changed between $x$ and $y$, which means that they never get underlined in $x$, hence they are $z$'s in $x'$. Thus, like previously, it is impossible that the first one is turned into a 0 while the second is turned into a 1.
        \item In $x$, these two bits are $1$'s. Then, the second bit gets turned into a 0 in $y$, which means that it never gets underlined in $x$. Like in the first case, we get that the first bit never gets underlined neither, once more making it impossible for these two bits to be turned respectively in 1 and 0.
    \end{enumerate}
    In all three cases, we get a contradiction. Thus, the corresponding bits in $x$ are also underlined at this step. Then, by induction, we get that at each step, the same bits are underlined in $x$ and $y$. Finally, we turn all the bits that are not underlined into $z$'s to get $x'$ and $y'$, hence $x' = y'$.
\end{proof}

\begin{remark}
We can define an equivalence relation $\sim$ over the subsets of $[2n]$ by saying that two subsets are equivalent if and only if they have the same Catalan string. \\
By combining Catalan factorization and Cover encodings, we can obtain a property-preserving encoding for $\sim$ on $\{0, 1\}^{2n}$.
We use this in \cref{sec:sperner}.
\end{remark}

\section{Erdős-Ko-Rado Theorem on Intersecting Families}
\label{sec:ErdosKoRado}
In this section, we define total search problems motivated by the well-known Erdős-Ko-Rado theorem on intersecting families and study their computational complexity.
First, we present a \PWPP-complete variant of the problem. 
Next, we modify the problem using a strong statement of the \EKR theorem to get a \PPP-complete variant.

Recall the definition of an intersecting family and the statement of the \EKR theorem.

\begin{definition}[Intersecting family]\label{def:if}
    Let $\Omega$ be any set. A family of sets $\F \subseteq 2^\Omega$ is an \emph{intersecting family} if no two sets are disjoint, i.e., if for any $A,B\in\F$, it holds that $A \cap B \neq \emptyset$. 
\end{definition}
\begin{classicaltheorem}[Erdős-Ko-Rado \cite{EKR}]\label{cthm:erdoskorado}
    Any intersecting family where each set has $k$ elements on a universe of size $m$ contains at most $\binom{m-1}{k-1}$ sets, and this bound is tight.
\end{classicaltheorem}

We start by defining a total search problem motivated by a special case of the \EKR theorem for families of $n$-sets in a universe of size $2n$ presented in the following corollary.

\begin{corollary}\label{cor:ekr} 
Any intersecting family where each set has $n$ elements on a universe of size $2n$ contains at most $\binom{2n-1}{n-1}$ sets, and this bound is tight. Furthermore, if $\F$ is an intersecting family of maximum size, then for every $n$-subset $S$, exactly one of $S$ and $\overline{S}$ is in $\F$.
\end{corollary}

Suppose that we have a collection, containing more than $\binom{2n-1}{n-1}$ sets of size $n$ on $2n$ elements.
Then, by \cref{cthm:erdoskorado}, there must be two sets that do not intersect.
This induces a total search problem of finding two such disjoint sets.
We consider an implicit representation of such a collection by a circuit $C$ whose inputs serve as indices in the collection.
The output of the circuit is a representation of the corresponding set as a characteristic vector of the $2n$ elements.
Of course, this representation does not guarantee that $C$ satisfies the conditions required for \cref{cthm:erdoskorado} to apply, which would make the problem not total; in this case, we allow evidence of this fact to be a solution to the problem. Namely, if for a given input $x$, we do not have $|C(x)| = n$, or two distinct indices $x,y$ represent the same set, i.e., $C(x)=C(y)$, we allow such inputs as solutions.

\begin{definition}[\WeakErdosKoRado]\label{def:problem_WeakEKR}
The problem $\WeakErdosKoRado$ is defined by the relation
\begin{description}
    \item[Instance:] A Boolean circuit $C\colon\{0,1\}^{\ceil{\log\left(\binom{2n-1}{n-1}\right)} + 1} \to \{0,1\}^{2n}$.
    \item[Solution:] One of the following:
    \begin{enumerate}[label=\roman*)]
    \item $x$ s.t. $|C(x)| \neq n$,
    \item $x \neq y$ s.t. $C(x) = C(y)$,
    \item $x, y$ s.t. $C(x) \cap C(y) = \emptyset$. 
    \end{enumerate}
\end{description}
\end{definition}

\noindent As we discussed in the introduction, the totality of this problem is proved using a ``weak" statement in extremal combinatorics, namely the first part of \cref{cor:ekr}, hence the name \textsc{Weak}. However, the analogy with $\WeakPigeon$ goes further. Indeed, our first main theorem is the following.

\begin{theorem}\label{thm:if_complete}
    $\WeakErdosKoRado$ is \PWPP-complete.
\end{theorem}

Throughout this section, we maintain $\alpha = \ceil{\log \binom{2n}{n}} = \ceil{\log \binom{2n-1}{n-1}} + 1$.
\begin{lemma}
    $\WeakErdosKoRado \in \PWPP$.
\end{lemma}
\begin{proof}
At a high level, we want to encode the sets using a shrinking circuit, in such a way that collisions correspond to disjoint sets. Observe that for $n$-sets in a universe of size $2n$, the only disjoint sets are complements, hence we get an equivalent instance of $\WeakErdosKoRado$ if we map each set to either itself or its complement, arbitrarily. In our construction, we map each set $S$ to the representative not containing 1.
That is, if $1\not\in S$, the set is left unchanged and, otherwise, it is mapped to its complement $\overline{S}$.
Note that by the pigeonhole principle, two sets that do not contain 1 must have a non-empty intersection since we work with $n$-subsets of $[2n]$.
To obtain a shrinking circuit, we make use of Cover encodings (\cref{sec:cover}) that give an optimal encoding of all $n$-sets by considering their lexicographic order.
Notice that if the input $S$ is not an $n$-set, we may map it arbitrarily to any $n$-set, as a collision, in this case, yields a solution to the $\WeakErdosKoRado$ instance. 

Formally, recall that we have $\Ecov : \{0,1\}^{2n} \rightarrow \{0,1\}^{\alpha}$ and $\Dcov : \{0,1\}^{\alpha} \rightarrow \{0,1\}^{2n}$. Now let $C : \{0,1\}^{\alpha} \rightarrow \{0,1\}^{2n}$ be an instance of $\IntersectingFamily$. We proceed to construct an instance $C' : \{0,1\}^{\alpha} \rightarrow \{0,1\}^{\alpha - 1}$ of $\WeakPigeon$ as follows:
$$
    C'(x) = \begin{cases} \Ecov(C(x)) & \text{if $C(x)_1=0$}\\ \Ecov(\overline{C(x)}) & \text{if $C(x)_1=1$} \end{cases}
$$
Note that since we only encode sets whose first bit is a $0$, by \cref{rmk:cover_0s}, we get that the first bit of the encoding always is a $0$, so we can consider only the $\ceil{\log(\binom{2n}{n})} - 1 = \alpha - 1$ last bits of $C'(x)$ for every $x$, which is why we say that $C'$ only outputs $\alpha - 1$ bits. Note also that if for some $x$, $C(x)$ does not have size $n$, then $\Ecov(C(x))$ and $\Ecov(\overline{C(x)})$ are still well-defined, even if they are meaningless. 

Now, suppose that we have a solution to $C'$, that is $x \neq y$ such that $C'(x) = C'(y)$. There are four cases to consider, depending on the first bits of $C(x),C(y)$. If $C(x)_1=C(y)_1=0$, then $\Ecov(C(x)) = C'(x) = C'(y) = \Ecov(C(y))$. If both $C(x)$ and $C(y)$ have size $n$, then by injectivity of $\Ecov$ on inputs of size $n$ (see \cref{lemma:cover_bij}), we get $C(x) = C(y)$, which is a solution to $\WeakErdosKoRado$. If one of them does not have size $n$, we also get a solution to $\WeakErdosKoRado$. The other cases are similar.
\end{proof}

\begin{remark}\label{rem:EKR_encoding}
Consider the circuit $E : \{0, 1\}^{2n} \rightarrow \{0, 1\}^{\alpha - 1}$, defined as follows. 
$$
    E(x) = \begin{cases}
                    0^{\alpha - 1} & \text{if $|x| \neq n$}\\
                    \Ecov(x) & \text{if $x_1=0$ and $|x| = n$}\\
                    \Ecov(\overline{x}) & \text{if $x_1=1$ and $|x| = n$}\\
            \end{cases}$$
Let $\mathcal{X} \subseteq \{0, 1\}^{2n}$ be the subset of $\{0, 1\}^{2n}$ corresponding to the $n$-subsets of $[2n]$. We define an equivalence relation $\sim$ on $\mathcal{X}$ by saying that two strings are equivalent if the corresponding subsets are either equal or disjoint. Note that this relation is transitive only because we work with $n$-subsets of $[2n]$. \\
Then, we have that $E$ is a property-preserving encoding for $\sim$ on $\mathcal{X}$.\\
Furthermore, the property that is preserved by $E$ is such that if two of its inputs collide, they form a solution to the problem we're interested in. \\
Then, to prove the inclusion of $\WeakErdosKoRado$ into $\PWPP$, it suffices to compose our instance of $\WeakErdosKoRado$ with $E$.
\end{remark}

\begin{lemma}\label{lemma:weakekr_hard}
    $\WeakErdosKoRado$ is $\PWPP$-hard.
\end{lemma}
\begin{proof}
Our goal is for the $\ErdosKoRado$ solver to find collisions in an instance $C'$ of \WeakPigeon. We use a variation of the graph hash product \cite{Krajicek05,c_ramsey}. The idea is to interpret the output of $C'$ as an index into the collection of all $n$-sets that do not contain 1. We then use the Cover decoding function to obtain a representation of the corresponding set, and by correctness of the encoding, any such set must have exactly $n$ elements -- and all the sets intersect since they do not contain 1. Hence, the only solutions to the $\WeakErdosKoRado$ instance are collisions, that yield solutions to the original circuit $C'$.

Formally, let $C' : \{0,1\}^{m} \rightarrow \{0,1\}^{m - 1}$ be an instance of $\WeakPigeon$. Let $n$ be the minimal integer such that $2^{m+1} \leq \binom{2n}{n}$. Then, $m+1 \leq \alpha$. We proceed to build a circuit $A : \{0, 1\}^{\alpha} \rightarrow \{0, 1\}^{\alpha - 2}$ whose size is polynomial in $m$ and such that from any collision in $A$ we can efficiently find a collision in $C'$. Recall that we have $\Ecov : \{0,1\}^{2n} \rightarrow \{0,1\}^{\alpha}$ and $\Dcov : \{0,1\}^{\alpha} \rightarrow \{0,1\}^{2n}$. We define $C: \{0, 1\}^\alpha \rightarrow \{0, 1\}^{2n}$ by 
$$
    C(x) = \Dcov(00 \cat A(x))
$$
By \cref{rmk:cover_0s}, since for every $x$, $(00 \cat A(x))$ is one of the $\binom{2n-1}{n-1}$ first possible inputs, we have that the set $\Dcov(00 \cat A(x))$ is an $n$-subset of $[2n]$ which does not contain the element 1. We observe that $C$ defines an instance of $\WeakErdosKoRado$. Now suppose that we have a solution to this instance. By correctness of the decoding, we can only have solutions of type iii), that is $x \neq y$ such that $C(x) = C(y)$. By injectivity of $\Dcov$ on its first $\binom{2n}{n}$ inputs (see \cref{lemma:cover_bij}), we get that $(00 \cat A(x)) = (00 \cat A(y))$ hence $A(x) = A(y)$ and from there we can retrieve a collision for $C'$.\qedhere
\end{proof}

\paragraph{\PPP-completeness using the tight bound}
We remark that \cref{cor:ekr} gives a \emph{tight} upper bound on the size of the collection. Furthermore, we know some structure of any collection whose size is exactly one $\binom{2n-1}{n-1}$: it must either not be an intersecting family, or it must contain either $[n]$ or $\overline{[n]}$. This is an example of a ``strong" theorem in extremal combinatorics. As discussed in the introduction, this observation allows us to modify the problem to be create a variant of $\WeakErdosKoRado$ that is to $\WeakErdosKoRado$ what $\Pigeon$ is to $\WeakPigeon$. The idea is to let $C$ encode a collection whose size exactly matches the threshold. We then let $C$ represent a collection of exactly $\binom{2n-1}{n-1}$ sets, and also allow preimages of $[n]$ and $\overline{[n]}$ as solutions. We show that modifying the problem in this manner makes it \PPP-complete, thus strengthening the analogy with $\Pigeon$. This technique is quite general, and we utilise it again in later sections.

\begin{definition}[\ErdosKoRado]\label{def:problem_EKR}
The problem $\ErdosKoRado$ is defined by the relation
\begin{description}
    \item[Instance:] A Boolean circuit $C\colon\{0,1\}^{\ceil{\log\left(\binom{2n-1}{n-1}\right)}} \to \{0,1\}^{2n}$.
    \item[Solution:] One of the following:
    \begin{enumerate}[label=\roman*)]
    \item $x$ s.t. $|C(x)| \neq n$ and $x < \binom{2n-1}{n-1}$,
    \item $x \neq y$ s.t. $C(x) = C(y)$  and $x, y < \binom{2n-1}{n-1}$,
    \item $x, y$ s.t. $C(x) \cap C(y) = \emptyset$ and $x, y < \binom{2n-1}{n-1}$,
    \item $x$ s.t. $C(x) = [n]$ or $\overline{[n]}$ and $x < \binom{2n-1}{n-1}$.
    \end{enumerate}
\end{description}
\end{definition}

\begin{theorem}\label{thm:if'_complete}
    $\ErdosKoRado$ is $\PPP$-complete.
\end{theorem}

\begin{lemma}\label{lemma:ekr_hard}
    $\ErdosKoRado$ is $\PPP$-hard.
\end{lemma}

\begin{proof}
This proof is similar in spirit to that of \cref{lemma:weakekr_hard}, except for some minor changes. The first one is that the instance of $\Pigeon$ might be a permutation, and thus not have collisions. We then need to be able to find the preimage of 0. This is done by solutions of type $iv)$. The second one is that we only look at the first $\binom{2n-1}{n-1}$ inputs of the $\Pigeon$ instance, so we have to modify it to make sure that all the possible solutions come from here. This is why we build the circuit $A$. 

Formally, let $C' : \{0,1\}^{m} \rightarrow \{0,1\}^{m}$ be an instance of \Pigeon, and let $n$ be the minimal integer such that $2^{m} < \binom{2n-1}{n-1}$. 
Since $\alpha = \ceil{\log \binom{2n-1}{n-1}} + 1$, we have $m < \alpha - 1$.
Define $A:\{0, 1\}^{\alpha-1} \rightarrow \{0, 1\}^{\alpha-1}$ by, 
$$
    A(x) = \begin{cases} C'(x) & \text{if $x < 2^m$}\\ x & \text{o.w.} \end{cases}
$$
It might be the case that the output of $A$ has less than $\alpha - 1$ bits, in which case we pad it with 0 on the left to make it an $(\alpha - 1)$-bit string. Recall that we have $\Ecov : \{0,1\}^{2n} \rightarrow \{0,1\}^{\alpha}$ and $\Dcov : \{0,1\}^{\alpha} \rightarrow \{0,1\}^{2n}$.

We proceed to build an instance $C : \{0,1\}^{\alpha-1} \rightarrow \{0,1\}^{2n}$ of $\ErdosKoRado$ by setting $C(x) = \Dcov(0 \cat A(x))$. Note that for any $x < \binom{2n-1}{n-1}$, we have $A(x) < \binom{2n-1}{n-1}$, thus $C(x) \subseteq [2n]$ is an $n$-subset and does not contain the element 1 by \cref{rmk:cover_0s}. 

Now, suppose that we have a solution to $C$. Since the index of a solution is $< \binom{2n-1}{n-1}$, the corresponding subset(s) must have size $n$ and can't contain $1$. If the solution is of the form $x, y$ such that $C(x) \cap C(y) = \emptyset$ then we have $|C(x) \ \cup \ C(y)| = |C(x)| + |C(y)| = 2n$ so we must have either $1 \in C(x)$ or $1 \in C(y)$, which is not possible.

Thus, any solution must be $x \neq y$ such that $C(x) = C(y)$ or $x$ such that $C(x) = [n]$ or $\overline{[n]}$. There are two cases to consider:
\begin{itemize}
    \item \textsc{Case} $\Dcov(0 \cat A(x)) = \Dcov(0 \cat A(y))$. Then $A(x) = A(y)$ since $\Dcov$ is injective on its first $\binom{2n}{n}$ inputs. But $C'$ has range $\subseteq [2^m-1]$ so any collision in $A$ must result from a collision in $C'$. Hence, we get that $x, y < 2^m$ give us a solution to $C'$.
    \item \textsc{Case} $\Dcov(A(x)) = [n]$ or $\overline{[n]}$. Since $A(x) < \binom{2n-1}{n-1}$ then $\Dcov(0 \cat A(x))$ does not contain element 1, so $C(x) = \overline{[n]} = \Dcov(0^{\alpha})$, thus $A(x) = 0^{\alpha - 1}$. This means that we have $x < 2^m$ and $x$ corresponds to a preimage of $0^m$ for $C'$.
\end{itemize}
In each case, we get a solution to our original problem.
\end{proof}

\begin{remark}
We often use that technique of creating a circuit $A$ from a circuit $C$, such that any collision (resp. preimage of 0) in $A$ must come from a collision (resp. preimage of 0) in $C$, and happen in the first inputs of $A$ (in the range where we want it to happen).
\end{remark}

\begin{lemma}
    $\ErdosKoRado \in \PPP$.
\end{lemma}

\begin{proof}
This proof is quite the same as the proof of \cref{lemma:weakekr_hard}, with two minor differences. The first one is that in the instance of $\Pigeon$ we create, there might be preimages of 0. These solutions to $\Pigeon$ correspond to solutions of type $iv)$ for $\ErdosKoRado$. The second difference is that we only perform the reduction on the first $\binom{2n-1}{n-1}$ inputs, and then map the others in such a way that they neither create a collision nor result in a preimage of 0.

Formally, suppose that we have an instance of $\ErdosKoRado$, i.e., a circuit $C : \{0,1\}^{\alpha-1} \rightarrow \{0,1\}^{2n}$. We proceed to construct an instance $C' : \{0,1\}^{\alpha-1} \rightarrow \{0,1\}^{\alpha-1}$ of $\Pigeon$ as follows:
$$
    C'(x) = \begin{cases}
                    \Ecov(C(x)) & \text{if $C(x)_1=0$ and $x < \binom{2n-1}{n-1}$}\\
                    \Ecov(\overline{C(x)}) & \text{if $C(x)_1=1$ and $x < \binom{2n-1}{n-1}$}\\
                     x & \text{if $x \geq \binom{2n-1}{n-1}$}
            \end{cases}
$$
In the case $x < \binom{2n-1}{n-1}$, since we only encode sets whose first bit is a $0$, by \cref{rmk:cover_0s}, we get that the first bit of the encoding  always is a $0$, so we can consider only the $\ceil{\log(\binom{2n}{n})} - 1 = \alpha - 1$ last bits of $C'(x)$ for every such $x$. Furthermore, if we consider the output of $\Ecov$ as an integer, we get that this integer is $< \binom{2n-1}{n-1}$ (because the set we encode is one of the first $\binom{2n-1}{n-1}$ in the lexicographic order). Note also that if for some $x$ such that $x < \binom{2n-1}{n-1}$, $C(x)$ does not have size $n$, then $C'(x)$ is still well-defined and less than $\binom{2n-1}{n-1}$, even if it is meaningless. 

Now, suppose that we have a solution to $C'$ of the form $x \neq y$ such that $C'(x) = C'(y)$. Again there are four cases to consider, depending on the first bits of $C(x),C(y)$. If $C(x)_1=C(y)_1=0$ then $\Ecov(C(x)) = C'(x) = C'(y) = \Ecov(C(y))$. If both $C(x)$ and $C(y)$ have size $n$, then by injectivity of $\Ecov$ on inputs of size $n$ (see \cref{lemma:cover_bij}), we get $C(x) = C(y)$, which is a solution to \ErdosKoRado. If one of them does not have size $n$, we also get a solution to \ErdosKoRado. The other cases are similar.

Now, suppose that we have a solution to $C'$ of the form $x$ such that $C'(x) = 0^{\alpha-1}$. Like previously, we get that $x < \binom{2n-1}{n-1}$. If $C(x)$ does not have size $n$ then $x$ is a solution. Now, suppose that $C(x)$ has size $n$. There are two cases to consider, depending on the first bit of $C'(x)$. If the first bit of $C(x)$ is 0, then, $\Ecov(C(x)) = 0^{\alpha}$ so $C(x) = 0^n \cat 1^n$ by \cref{eq:decode_0_string} and \cref{lemma:cover_bij}. Thus, $C(x) = \overline{[n]}$. Instead, if the first bit of $C(x)$ is 1, then $\Ecov(\overline{C(x)}) = 0^{\alpha}$ so $\overline{C(x)} = \overline{[n]}$ and thus $C(x) = [n]$. In either case, we get a solution to our original problem.
\end{proof}

\begin{remark}
Like previously, the idea behind that proof is to compose our instance of $\ErdosKoRado$ with the property-preserving encoding we defined in \cref{rem:EKR_encoding}. However, this time it is not only the collisions that are of interest to us, but also the preimages of the 0 string.
\end{remark}

\subsection{A Generalized Erd\H{o}s-Ko-Rado Problem} 
For the previous problems, we were only considering a very restricted version of the Erd\H{o}s-Ko-Rado theorem, namely for an intersecting family of $n$-subsets of $[2n]$. We now consider a more general version where we consider an intersecting family of $n$-subsets of $[kn]$ for some $k > 2$. 

We now fix some $k > 2$ for the rest of this section. The Erd\H{o}s-Ko-Rado theorem states that if $\F$ is an intersecting family where each set has $n$ elements on a universe of size $kn$, then $\F$ contains at most $\binom{kn-1}{n-1}$ sets. Then, we can define the following $\TFNP$ problem, very similar to $\WeakErdosKoRado$.

\begin{definition}[\WeakGeneralErdosKoRado]\label{def:problem_WeakGeneralEKR}
The problem $\WeakGeneralErdosKoRado$ is defined by the relation
\begin{description}
    \item[Instance:] A Boolean circuit $C\colon\{0,1\}^{\ceil{\log\left(\binom{kn-1}{n-1}\right)} + 1} \to \{0,1\}^{kn}$.
    \item[Solution:] One of the following:
    \begin{enumerate}[label=\roman*)]
    \item $x$ s.t. $|C(x)| \neq n$,
    \item $x \neq y$ s.t. $C(x) = C(y)$,
    \item $x, y$ s.t. $C(x) \cap C(y) = \emptyset$. 
    \end{enumerate}
\end{description}
\end{definition}

\begin{proposition}\label{prop:weakgenekr_hard}
    $\WeakGeneralErdosKoRado$ is $\PWPP$-hard.
\end{proposition}

\begin{proof}
This proof is very similar to the proof of \cref{lemma:weakekr_hard}, except that instead of working with $n$-subsets of $[2n]$, we work with $n$-subsets of $[kn]$. There is also a technical change, which is that this time we work with $n$-subsets of $[kn]$ that \emph{do} contain the element 1. This is necessary to make sure that we have an intersecting family, but it adds some more technicality. For the same reason, we need $A$ to shrink more than in the previous proof. However, the idea behind the proof is exactly the same, with the same use of the graph-hash product on a large intersecting family.

Formally, let $C' : \{0,1\}^{m} \rightarrow \{0,1\}^{m - 1}$ be an instance of $\WeakPigeon$. Let $n$ be the minimal integer such that $2^{m+1} \leq \binom{kn}{n}$. Now, let $\alpha = \ceil{\log \binom{kn}{n}}$. Then, $m+1 \leq \alpha$. We also define $a = \ceil{\log(k)}$. By definition of $\alpha$, we have $\binom{kn}{n} \geq 2^{\alpha - 1}$. We also have $\frac{1}{k} \geq \frac{1}{2^{a}}$, so $\binom{kn-1}{n-1} = \frac{1}{k}\binom{kn}{n} \geq \frac{2^{\alpha - 1}}{k} \geq 2^{\alpha - 1 - a}$. Like in the proof of \cref{lemma:weakcayley_hard}, we can build a circuit $A' : \{0, 1\}^{\alpha} \rightarrow \{0, 1\}^{\alpha - 1 - a}$ whose size is polynomial in $m$ and such that from any collision in $A'$ we can efficiently find a collision in $C'$. Let $s \in \{0, 1\}^{\alpha}$ be the binary encoding on $\alpha$ bits of $\binom{kn}{n} - \binom{kn-1}{n-1}$. We use the Cover encoding functions for $n$-subsets of $[kn]$: $\Ecov^{n, kn} : \{0,1\}^{kn} \rightarrow \{0,1\}^{\alpha}$ and $\Dcov^{n, kn} : \{0,1\}^{\alpha} \rightarrow \{0,1\}^{kn}$. 

We define $C: \{0, 1\}^{\alpha} \rightarrow \{0, 1\}^{kn}$ by $C(x) = \Dcov^{k, kn}(s \oplus 0^{a+1} \cat A'(x))$. For every $x$, we have that $(0^{a+1} \cat A'(x))$ is one of the first $2^{\alpha - 1-a}$ elements of $\{0, 1\}^{\alpha}$ in the lexicographic order, hence it is one of the first $\binom{kn-1}{n-1}$ first. Thus, the rank of $s \oplus 0^{a+1} \cat A'(x)$ in the lexicographic order is between $\binom{kn}{n} - \binom{kn-1}{n-1}$ and $\binom{kn}{n} - 1$ counting from 0. The last $\binom{kn - 1}{n - 1}$ $n$-subsets of $[kn]$ in the lexicographic order correspond to subsets that contain the element 1. Hence, for every $x$, we have that the set $\Dcov^{n, kn}(s \oplus 0^{1 + a} \cat A'(x))$ is an $n$-subset of $[kn]$ which contains the element 1. We observe that $C$ defines an instance of $\WeakGeneralErdosKoRado$. 

Now, suppose that we have a solution to this instance. We consider each solution type separately.
\begin{enumerate}
    \item[i)] It cannot be $x$ such that $|C(x)| \neq n$ because $C(x) = \Dcov^{n, kn}(s \oplus 0^{1 + a} \cat A'(x))$ is an $n$-subset of $[kn]$.
    \item[ii)] By the previous, $1 \in C(x)$ and $1 \in C(y)$ so $1 \in C(x) \cup C(y)$, which is a contradiction.
\item[iii)] By injectivity of $\Dcov^{n, kn}$ on its first $\binom{kn}{n}$ inputs (see \cref{lemma:cover_bij}), we get that $(s \oplus 0^{1 + a} \cat A'(x)) = (s \oplus 0^{1 + a} \cat A'(y))$ hence $A'(x) = A'(y)$ and from there we can retrieve a collision for $C'$.\qedhere
\end{enumerate}
\end{proof}

To prove that $\WeakGeneralErdosKoRado \in \PWPP$, we present some useful definitions and results related to the \EKR theorem.

\begin{definition}
    If $k$ divides $m$, a $(k, m)$-\emph{parallel class} is a set of $m/k$ $k$-subsets of $[m]$ which partition $[m]$.
\end{definition}

\begin{classicaltheorem}[Baranyai, \cite{Baranyai}] \label{cthm:baranyai} 
    If $k$ divides $m$, we can define $\binom{m - 1}{k - 1}$ $(k, m)$-parallel classes $\mathcal{A}_1, \ldots, \mathcal{A}_{\binom{m-1}{k-1}}$ such that each $k$-subset of $[m]$ appears in exactly one $\mathcal{A}_i$.
\end{classicaltheorem}

\begin{remark}
Note that this result proves the \EKR theorem in the case where the size of the subsets divides the size of the universe. \\
Note also that up to renaming the elements, we can assume that $\mathcal{A}_1$ consists exactly of the sets $\{1, 2, \ldots, n\}, \{n+1, n+2, \ldots, 2n\}, \ldots$,  and $\{(k-1)n+1, (k-1)n+2, \ldots, kn\}$.
\end{remark}
However, all known proofs of this theorem are inefficient, in the sense that there is no known way to define $\mathcal{A}_1, \ldots, \mathcal{A}_{\binom{m-1}{k-1}}$ such that given a $k$-subset of $[m]$, we can find in polynomial time the only $i$ such that this subset appears in $\mathcal{A}_i$.  We make this assumption explicit. 

\begin{assumption}[efficient Baranyai assumption]\label{assump:baranyai}
    There is an efficient procedure to define $\mathcal{A}_1, \ldots, \mathcal{A}_{\binom{m-1}{k-1}}$ and a circuit $Bar : \{0, 1\}^m \rightarrow [\binom{m - 1}{k - 1}]$ which takes as input a $k$-subset of $[m]$ and returns the only index $i$ such that this subset appears in $\mathcal{A}_i$. Furthermore, we assume that $\mathcal{A}_1$ consists exactly of the sets $\{1, 2, \ldots, n\}, \{n+1, n+2, \ldots, 2n\}, \ldots$,  and $\{(k-1)n+1, (k-1)n+2, \ldots, kn\}$.
\end{assumption}

\begin{proposition}\label{prop:weakgenekr_pwpp}
    Under \cref{assump:baranyai}, $\WeakGeneralErdosKoRado\in\PWPP$.
\end{proposition}

\begin{proof}
At a high level, the proof goes as follows. We are given strictly more than $\binom{kn-1}{n-1}$ subsets of $[kn]$. We map them to elements of $[\binom{kn-1}{n-1}]$ in the following way. If one set does not have size $n$, we map it anywhere. If it has size $n$, we map it to the only $i$ such that the set is in $\mathcal{A}_i$. This defines an instance of $\WeakPigeon$. In any collision for this instance, we must have either a set that does not have size $n$, or two sets in the same parallel class, which means that either they are equal, or they do not intersect.

Formally, by assumption, we have a circuit $Bar : \{0, 1\}^{kn} \rightarrow [\binom{kn - 1}{n - 1}]$ which takes as input an $n$-subset of $[kn]$ and returns the only index $i$ such that this subset appears in $\mathcal{A}_i$. We define a circuit $Bar' : \{0, 1\}^{kn} \rightarrow \{0, 1\}^{\ceil{\log\binom{kn-1}{n-1}}}$ which takes as input an $n$-subset of $[kn]$ and returns the binary encoding on $\ceil{\log\binom{kn-1}{n-1}}$ bits of the only index $i$ such that this subset appears in $\mathcal{A}_i$. Now, suppose that we have an instance $C : \{0,1\}^{\ceil{\log\left(\binom{kn-1}{n-1}\right)} + 1} \rightarrow \{0,1\}^{kn}$ of $\WeakGeneralErdosKoRado$. We set $C' = Bar' \circ C$. Then, we have $C' : \{0,1\}^{\ceil{\log\left(\binom{kn-1}{n-1}\right)} + 1} \rightarrow \{0, 1\}^{\ceil{\log\binom{kn-1}{n-1}}}$ so $C'$ is an instance of $\WeakPigeon$. 

Now, suppose that we have a solution to this instance of $\WeakPigeon$, that is $x \neq y \in \{0,1\}^{\ceil{\log\left(\binom{kn-1}{n-1}\right)} + 1}$ such that $C'(x) = C'(y)$. Then, $Bar'(C(x)) = Bar'(C(y))$. If one of $C(x), C(y)$ does not have size $n$, we have a solution to our instance of $\WeakGeneralErdosKoRado$, and similarly if $C(x) = C(y)$. Otherwise, it means that $C(x), C(y)$ are distinct $n$-subsets of $[kn]$ that appear in the same $(n, kn)$-parallel class. By definition of a parallel class, it means that these 2 sets are part of a partition of $[kn]$, hence they don't intersect and they form a solution to our original instance of $\WeakGeneralErdosKoRado$.
\end{proof}

\begin{remark}
Let $\mathcal{X}$ be the set of $n$-subsets of $[kn]$. We define an equivalence relation $\sim$ on $\mathcal{X}$ by saying that two $n$-subsets $X$ and $Y$ of $[kn]$ are equivalent if and only $Bar(X) = Bar(Y)$, meaning that they are in the same $(n, kn)$-parallel class in the partition induced by $Bar$. \\
Then, we have that $Bar$ is a property-preserving encoding for $\sim$ on $\mathcal{X}$.\\
Note that two equivalent subsets are either equal or disjoint. Hence, the property that is preserved by $Bar$ is such that if two of its inputs collide, they form a solution to our problem. \\
Then, to prove the inclusion of $\WeakGeneralErdosKoRado$ into \PPP, it suffices to compose our instance of \WeakGeneralErdosKoRado with $Bar$.
\end{remark}

The previous two propositions establish the following result.
\begin{theorem}\label{thm:weakgeneralekr}
    Under \cref{assump:baranyai}, \WeakGeneralErdosKoRado is \PWPP-complete.
\end{theorem}

\paragraph{\PPP-completeness using the tight bound} Like for the case of $n$-subsets of $[2n]$, we can define a ``tight" version of the previous problem, which is very similar to $\ErdosKoRado$.

\begin{definition}[\GeneralErdosKoRado]\label{def:problem_GeneralEKR}
The problem $\GeneralErdosKoRado$ is defined by the relation
\begin{description}
    \item[Instance:] A Boolean circuit $C\colon\{0,1\}^{\ceil{\log\left(\binom{kn-1}{n-1}\right)}} \to \{0,1\}^{kn}$.
    \item[Solution:] One of the following:
    \begin{enumerate}[label=\roman*)]
    \item $x$ s.t. $|C(x)| \neq n$ and $x < \binom{kn-1}{n-1}$,
    \item $x \neq y$ s.t. $C(x) = C(y)$  and $x, y < \binom{kn-1}{n-1}$,
    \item $x, y$ s.t. $C(x) \cap C(y) = \emptyset$ and $x, y < \binom{kn-1}{n-1}$,
    \item $x$ s.t. $C(x) = \{1, 2, \ldots, n\}$ or $\{n+1, n+2, \ldots, 2n\}$, or..., or $\{(k-1)n+1, (k-1)n+2, \ldots, kn\}$ and $x < \binom{kn-1}{n-1}$.
    \end{enumerate}
\end{description}
\end{definition}

First, let's see why this problem is total. Suppose that we have a list of $\binom{kn-1}{n-1}$ subsets of $[kn]$. If one of the sets does not have $n$ elements, if two of the sets are equal, or if two of the sets don't intersect, we have a solution. Now, suppose that we have an intersecting family of $\binom{kn-1}{n-1}$ distinct $n$-subsets of $[kn]$. \\
Now, consider a collection of $(n, kn)$-parallel classes $\mathcal{A}_1, \ldots, \mathcal{A}_{\binom{kn-1}{n-1}}$ such that each $n$-subset of $[kn]$ appears in exactly one $\mathcal{A}_i$ (which exists by \cref{cthm:baranyai}). Up to renaming the elements, we can assume that $\mathcal{A}_1$ is composed of the $k$ $n$-subsets $\{1, 2, \ldots, n\}$, $\{n+1, n+2, \ldots, 2n\}$, ... and $\{(k-1)n+1, (k-1)n+2, \ldots, kn\}$.\\
Since we have an intersecting family of distinct subsets, no two subsets can be in the same $\mathcal{A}_i$, and we have as many subsets as $\mathcal{A}_i$'s, which means that one of the subsets is in $\mathcal{A}_1$, hence that it is one of the particular subsets we are looking for. This proves that $\GeneralErdosKoRado \in \TFNP$. We then have the following result.

\begin{proposition}
    $\GeneralErdosKoRado$ is $\PPP$-hard.
\end{proposition}

\begin{proof}
Informally, this proof is very much like the proof of \cref{prop:weakgenekr_hard}, with the same technicalities as in the proof of \cref{lemma:ekr_hard}. The idea is again to interpret the outputs of an instance of $\Pigeon$ as indices into the collection of all the $n$-subsets of $[kn]$ which contain the element 1. Solutions of type $iv)$ correspond to preimages of 0.
Like for \cref{lemma:ekr_hard}, we need to define $A$ to make sure that all solutions to our instance of $\GeneralErdosKoRado$ indeed come from the instance of $\Pigeon$.

Formally, let $C' : \{0,1\}^{m} \rightarrow \{0,1\}^{m}$ be an instance of \Pigeon, and let $n$ be the minimal integer such that $2^{m} \leq \binom{kn-1}{n-1}$. We set $\alpha = \ceil{\log \binom{kn}{n}}$ and $\beta = \ceil{\log\binom{kn-1}{n-1}} + 1$. Then, $\beta - 1 \geq m$. Define $A:\{0, 1\}^{\beta - 1} \rightarrow \{0, 1\}^{\beta-1}$ by, 
$$
    A(x) = \begin{cases} C'(x) & \text{if $x < 2^m$}\\ x & \text{if $x \geq 2^m$} \end{cases}
$$
It might be the case that the output of $A$ has less than $\beta - 1$ bits, in which case we pad it with 0 on the left to make it an $(\beta - 1)$-bit string. Let $s \in \{0, 1\}^{\alpha}$ be the binary encoding on $\alpha$ bits of $\binom{kn}{n} - 1$. Recall that we have $\Ecov^{n, kn} : \{0,1\}^{kn} \rightarrow \{0,1\}^{\alpha}$ and $\Dcov^{n, kn} : \{0,1\}^{\alpha} \rightarrow \{0,1\}^{kn}$. 

We proceed to build an instance $C : \{0,1\}^{\beta-1} \rightarrow \{0,1\}^{kn}$ of $\GeneralErdosKoRado$ by setting $C(x) = \Dcov^{n, kn}(s - 0^{\alpha + 1 - \beta} \cat A(x))$ where - represents the subtraction in binary (mod $2^{\alpha}$). For every $x < \binom{kn-1}{n-1}$, we have that $(0^{\alpha + 1 - \beta} \cat A(x))$ is one of the first $\binom{kn-1}{n-1}$ elements of $\{0, 1\}^{\alpha}$ in the lexicographic order. Thus, the rank of $s - 0^{\alpha+1-\beta} \cat A(x)$ in the lexicographic order is between $\binom{kn}{n} - \binom{kn-1}{n-1}$ and $\binom{kn}{n} - 1$ counting from 0. The last $\binom{kn - 1}{n - 1}$ $n$-subsets of $[kn]$ in the lexicographic order correspond to subsets that contain the element 1. Hence, for every $x < \binom{kn-1}{n-1}$, we have that the set $\Dcov^{n, kn}(s - 0^{\alpha+1-\beta} \cat A(x))$ is an $n$-subset of $[kn]$ which contains the element 1. We observe that $C$ defines an instance of $\GeneralErdosKoRado$.

Now, suppose that we have a solution to this instance. We consider each solution type separately.

\begin{enumerate}
    \item[i)] It cannot be $x$ such that $|C(x)| \neq n$ because $C(x) = \Dcov^{n, kn}(s - 0^{\alpha + 1 - \beta} \cat A(x))$ is an $n$-subset of $[kn]$.
    \item[ii)] By the previous, $1 \in C(x)$ and $1 \in C(y)$ so $1 \in C(x) \cup C(y)$, which is a contradiction.
    \item[iii)] By injectivity of $\Dcov^{n, kn}$ on its first $\binom{kn}{n}$ inputs (see \cref{lemma:cover_bij}), we get that $(s - 0^{\alpha + 1 - \beta} \cat A(x)) = (s - 0^{\alpha +  1 - \beta} \cat A(y))$ hence $A(x) = A(y)$ and from there we can retrieve a collision for $C'$ by design of $A$.
    \item[iv)] If it is $x$ such that $C(x)$ is one of the $k$ particular subsets we're looking for, since we know that $1 \in C(x)$, it means that $C(x) = [n]$. When we consider $n$-subsets of $[kn]$, the characteristic vector of $[n]$ is the last one in the lexicographic order, which means that $[n] = \Dcov^{n, kn}(s)$. Furthermore, $[n] = C(x) = \Dcov^{n, kn}(s - 0^{\alpha + 1 - \beta} \cat A(x))$, the rank of $s - 0^{\alpha+1-\beta} \cat A(x)$ in the lexicographic order is between $\binom{kn}{n} - \binom{kn-1}{n-1} + 1$ and $\binom{kn}{n}$ and $\Dcov^{n, kn}$ is injective on its first $\binom{kn}{n}$ inputs. Thus, $s - 0^{\alpha + 1 - \beta} \cat A(x) = s$, which implies that $A(x) = 0$. By definition of $A$, this can only mean that $C'(x) = 0^m$.
\end{enumerate}
In either case, we get a solution to our original problem.
\end{proof}

\begin{proposition}
    Under \cref{assump:baranyai}, $\GeneralErdosKoRado\in\PPP$.
\end{proposition}

\begin{proof}
The proof of this result resembles a lot the proof of \cref{prop:weakgenekr_pwpp}. The idea is the same: we are given $\binom{kn-1}{n-1}$ subsets of $[kn]$. We map each of them to an element of $[\binom{kn-1}{n-1}$ as follows. If a set does not have $n$ elements, we map it anywhere, and if it has $n$ elements, we map it to the only $i$ such that this set is in $\mathcal{A}_i$. This defines an instance of $\Pigeon$. If we have a collision, it results in a solution like before. If we have a preimage of 0, it is a set in $\mathcal{A}_1$, which means it is one of the sets we are looking for. 
The definition of $C'$ has some technicality since we need to take care of the last inputs to make sure that they are not involved in a collision or result in a preimage of 0.

More formally, we have by assumption a circuit $Bar : \{0, 1\}^{kn} \rightarrow [\binom{kn - 1}{n - 1}]$ which takes as input an $n$-subset of $[kn]$ and returns the only index $i$ such that this subset appears in $\mathcal{A}_i$. We define a circuit $Bar' : \{0, 1\}^{kn} \rightarrow \{0, 1\}^{\ceil{\log\binom{kn-1}{n-1}}}$ which takes as input an $n$-subset of $[kn]$ and returns the binary encoding on $\ceil{\log\binom{kn-1}{n-1}}$ bits of $i-1$ where $i$ is the only index such that this subset appears in $\mathcal{A}_i$. \\
Now, suppose that we have an instance $C : \{0,1\}^{\ceil{\log\left(\binom{kn-1}{n-1}\right)}} \rightarrow \{0,1\}^{kn}$ of $\WeakGeneralErdosKoRado$. \\
We set $$
    C'(x) = \begin{cases} Bar' \circ C(x) & \text{if $x < \binom{kn-1}{n-1}$}\\ x & \text{if $x \geq \binom{kn-1}{n-1}$} \end{cases}
$$Then, we have $C' : \{0,1\}^{\ceil{\log\left(\binom{kn-1}{n-1}\right)}} \rightarrow \{0, 1\}^{\ceil{\log\binom{kn-1}{n-1}}}$ so $C'$ is an instance of $\Pigeon$. \\
Now, suppose that we have a solution to this instance of $\Pigeon$. There are two cases to consider. \begin{enumerate}
    \item It is $x \neq y \in \{0,1\}^{\ceil{\log\left(\binom{kn-1}{n-1}\right)}}$ such that $C'(x) = C'(y)$. By construction of $C'$ (and by definition of $Bar'$), this means that $x, y < \binom{kn-1}{n-1}$. We have $Bar'(C(x)) = Bar'(C(y))$. If one of $C(x), C(y)$ does not have size $n$, we have a solution to our instance of $\GeneralErdosKoRado$, and similarly if $C(x) = C(y)$. Otherwise, it means that $C(x), C(y)$ are distinct $n$-subsets of $[kn]$ that appear in the same $(n, kn)$-parallel class. By definition of a parallel class, it means that these 2 sets are part of a partition of $[kn]$, hence they don't intersect and they form a solution to our original instance of $\GeneralErdosKoRado$.
    \item It is $x$ such that $C'(x) = 0^{\ceil{\log\left(\binom{kn-1}{n-1}\right)}}$. By construction of $C'$, it means that $x < \binom{kn-1}{n-1}$. We have $Bar'(C(x)) = 0^{\ceil{\log\left(\binom{kn-1}{n-1}\right)}}$. If $C(x)$ does not have size $n$, it is a solution to our original instance. If it has size $n$, it means that it is an $n$-subset of $[kn]$ which is in $\mathcal{A}_1$. By assumption, the only such subsets are the particular ones we're looking for. Hence, $x$ is a solution to our original instance of $\GeneralErdosKoRado$.\qedhere
\end{enumerate} 
\end{proof}

\begin{remark}
As before, the idea behind that proof is to compose our instance of $\GeneralErdosKoRado$ with the property-preserving encoding $Bar$. However, this time it is not only the collisions that are of interest to us, but also the preimages of the 0 string.
\end{remark}

The previous two propositions establish the following result.
\begin{theorem}\label{thm:generalekr}
    Under \cref{assump:baranyai}, $\GeneralErdosKoRado$ is \PPP-complete.
\end{theorem}

\section{Sperner's Theorem on Largest Antichains} \label{sec:sperner}
\label{sec:Sperner}
We now turn our attention to a different existence theorem from extremal combinatorics, concerning antichains. We say a family of sets $\F\subseteq2^\Omega$ is an \emph{antichain} if for every $A \neq B \in \F$, it holds that $A \not\subseteq B$. A well-known theorem by Sperner gives a characterization of the largest antichain. As before, for an appropriate input size, this induces a total search problem of finding two distinct sets $A,B$ for which $A \subseteq B$. As in the previous section, we consider both a weak and a strong version, and prove the weak version to be \PWPP-complete, and the strong one \PPP-complete.

\begin{classicaltheorem}[Sperner \cite{sperner}]\label{cthm:sperner}
    The largest antichain on any universe of $2n$ elements is unique and consists of all subsets of size $n$.
\end{classicaltheorem}

Like before, we consider an implicit representation of the collection of subsets via a circuit $C$ whose input corresponds to an index into the collection, and whose output is the characteristic vector of the corresponding set.

\begin{definition}[\WeakSperner]\label{def:problem_WeakSperner}
The problem $\WeakSperner$ is defined by the relation
\begin{description}
    \item[Instance:] A Boolean circuit $C\colon\{0,1\}^{\ceil{\log\left(\binom{2n}{n}\right)} + 1} \to \{0,1\}^{2n}$.
    \item[Solution:] $x \neq y$ s.t. $C(x) \subseteq C(y)$.
\end{description}
\end{definition}

\begin{theorem} \label{thm:weaksperner_complete}
$\WeakSperner$ is \PWPP-complete
\end{theorem}

For the rest of this section, we set $\alpha = \ceil{\log \binom{2n}{n}} = \ceil{\log \binom{2n-1}{n-1}} + 1$.

\begin{lemma}\label{lemma:weaksperner_hard}
    $\WeakSperner$ is \PWPP-hard.
\end{lemma}
\begin{proof}
We explain the reduction at a high level. We reduce from $\WeakErdosKoRado$ and create an instance of $\WeakSperner$ by including each set from the $\WeakErdosKoRado$ instance, as well as its complement. If we find a solution to $\WeakSperner$, one of the sets must be contained within another. If one of the two sets does not have size $n$, we obtain a solution to $\WeakErdosKoRado$ of type i). Otherwise, the duplicated sets must be equal, and hence the original sets are either equal, or one of the sets is the complement of the other.

Formally, suppose that we have an instance $C : \{0,1\}^{\alpha} \rightarrow \{0,1\}^{2n}$ of $\WeakErdosKoRado$. Write $x = yb$ where $b$ is a bit. We build an instance $C' : \{0,1\}^{\alpha + 1} \rightarrow \{0,1\}^{2n}$ of $\Sperner$ as follows. 
$$
    C'(x) = 
    \begin{cases}
        C(y) & \text{if $b=0$}\\
        \overline{C(y)} & \text{if $b=1$}
    \end{cases}
$$
 Now, suppose that we have a solution to this instance of $\Sperner$, that is $x \neq x'$ such that $C'(x) = C'(x')$. Write $x = yb$ and $x'=y'b'$. There are four cases to consider. If $b=b'=0$. Then $y \neq y'$ and $C(y) = C'(x) \subseteq C'(x') = C(y')$. If $C(y)$ and $C(y')$ both have size $n$, then $C(y) = C(y')$, and if this is not the case we get a solution for $C$. In both cases, we get a solution for $\WeakErdosKoRado$. The other cases are similar; in all four cases, we get a solution to our original problem, so $\WeakSperner$ is \PWPP-hard.
\end{proof}

\begin{classicaltheorem}[Dilworth's Theorem, \cite{dilworth}]\label{lemma:dilworth_thm}
    The size of the largest antichain in $(2^{[2n]}, \subseteq)$ is equal to the size of the smallest chain partition, namely $\binom{2n}{n}$.
\end{classicaltheorem}

\begin{lemma}
    $\WeakSperner \in \PWPP$.
\end{lemma}

\begin{proof}
We give a high-level overview of the reduction from $\WeakSperner$ to $\WeakPigeon$. 

Fix an arbitrary partition into chains of $(2^{[2n]}, \subseteq)$ of size $\binom{2n}{n}$ (which exists by \cref{cthm:sperner,lemma:dilworth_thm}). 
Since we have more than $\binom{2n}{n}$ inputs in an instance of $\WeakSperner$, by the pigeonhole principle, two distinct inputs must end up in the same chain. 
We want to give an identifier to each of these chains, using $\alpha$ bits, such that for any subset we are be able to quickly find the identifier of the chain to which it belongs. 
To do so, in each chain, we choose as representative the $n$-subset of the chain, that is guaranteed to exist by \cref{lemma:dilworth_thm}. Then, the identifier of the chain is the Cover encoding on this subset. To map a subset to the representative of its chain, we make use Catalan factorizations (\cref{sec:catalan}). Once we have this, from each subset we can efficiently get the $n$-subset in its chain and therefore the identifier of the chain. Finally, a collision in the identifiers is equivalent to two elements in the same chain, which means a solution for $\WeakSperner$.

Formally, let $C: \{0, 1\}^{\alpha + 1} \rightarrow \{0, 1\}^{2n}$ be an instance of $\WeakSperner$. We proceed to construct an instance of $\WeakPigeon$ as follows: if $x \in \{0, 1\}^{\alpha + 1}$, we have $X:=C(x) \in \{0, 1\}^{2n}$ which represents a subset of $[2n]$. Let $(X', k) = \Ecat(X)$ be the Catalan factorization of $X$, $l$ be the number of $z$'s in $X'$ and $m$ the number of bits underlined during the construction of $X'$. Note that every time we underline bits we underline simultaneously a 0 and a 1, thus $m$ is even. Then, $l = 2n - m$ is an even number. Now, let $S(x) = \Dcat^{(l/2)}(X')$. Then, since $X'$ has the same number of 1's and 0's and since we replaced half of the $z$'s by 1's and the other half by 0's, we have that $S(x)$ represents an $n$-subset of $[2n]$. Informally, it is the $n$-subset of the chain that contains $X$, and replacing $z$'s by 1's enables us to move inside that chain. Finally, we set $C'(x) = \Ecov(S(x)) \in \{0, 1\}^{\alpha}$. We observe that $C'$ is an instance of $\WeakPigeon$. 

Now suppose that we have a solution to this instance of $\WeakPigeon$, that is $x \neq y$ such that $C'(x) = C'(y)$. Then, by injectivity of $\Ecov$ on the $n$-subsets of $[2n]$ (see \cref{lemma:cover_bij}), we get that $S(x) = S(y)$. Informally, this means that $C(x)$ and $C(y)$ belong to the same chain and thus that one is contained is the other. Let's now prove it formally. Let $(X', k) = \Ecat(X) = \Ecat(C(x))$ be the Catalan factorization of $X$ and $l$ be the number of $z$'s in $X'$, and let $(Y', k') = \Ecat(Y) = \Ecat(C(y))$. We have $S(x) = \Dcat(X', l/2)$ so by \cref{lemma:catalan_id2}, the Catalan string that corresponds to $S(x)$ is $X'$. Similarly, the Catalan string that corresponds to $S(y)$ is $Y'$. Since $S(x) = S(y)$, we get $X' = Y'$. We have that $X = \Dcat(\Ecat(X))$ and that $Y = \Dcat(\Ecat(Y))$ by \cref{lemma:cat_bij}, so $X = \Dcat(X', k)$ and $Y = \Dcat(Y', k') = \Dcat(X', k')$. By symmetry of $x$ and $y$ we can assume that $k \leq k'$. Then, to go from $X'$ to $X$ we added $k$ elements (the ones corresponding to the last $k$ $z$'s in $X'$) while to go from $X'$ to $Y$ we added these same $k$ elements plus $k'-k$ others. Hence, $C(x) = X \subseteq Y = C(y)$.
\end{proof}

\begin{remark}\label{rmk:Sperner_encoding}
Consider the circuit $E : \{0, 1\}^{2n} \rightarrow \{0, 1\}^{\alpha}$, defined as follows. On input $X \in \{0, 1\}^{2n}$, it computes $(X', k)$ the Catalan factorization of $X$, $l$ the number of $z$ in $X'$. Then, it computes $S(X) = \Dcat^{(l/2)}(X')$ and finally returns $\Ecov(S(X))$. \\
Let $\mathcal{X} = 2^{[2n]}$. We define an equivalence relation on $\mathcal{X}$ by saying that two subsets are equivalent if and only if they have the same Catalan string. \\
Then, we showed in the previous proof that $E$ is a property-preserving encoding for $\sim$ on $\mathcal{X}$.
Note that we also showed that if we have two equivalent subsets, one is included in the other. Hence, the property that is preserved by $E$ is such that if two of its inputs collide, they form a solution to our problem. \\
Then, to prove the inclusion of $\WeakSperner$ into $\PWPP$, it suffices to compose our instance of $\WeakSperner$ with $E$.
\end{remark}

\paragraph{\PPP-completeness using the tight bound} As with Erdős-Ko-Rado, we observe that the bound in theorem is {tight}, and we know the unique antichain of size $\binom{2n}{n}$, so we have some structural information about any collection of size $\binom{2n}{n}$. From that strong theorem, employing the same technique as before, we modify the problem to let the circuit represent a collection of that exact size. By \cref{cthm:sperner}, we observe that if $\mathcal{F}$ is an antichain with $|\mathcal{F}| = \binom{2n}{n}$, then $\mathcal{F}$ must contain $\overline{[n]}$. This leads us to define the following problem.

\begin{definition}[\Sperner]\label{def:problem_Sperner}
The problem $\Sperner$ is defined by the relation
\begin{description}
    \item[Instance:] A Boolean circuit $C\colon\{0,1\}^{\ceil{\log\left(\binom{2n}{n}\right)}} \to \{0,1\}^{2n}$.
    \item[Solution:] One of the following:
    \begin{enumerate}[label=\roman*)]
    \item $x \neq y$ s.t. $C(x) \subseteq C(y)$ and $x, y < \binom{2n}{n}$,
    \item $x$ s.t. $C(x) = \overline{[n]}$ and $x < \binom{2n}{n}$. 
    \end{enumerate}
\end{description}
\end{definition}

\begin{theorem} \label{thm:sperner'_complete}
$\Sperner$ is \PPP-complete.
\end{theorem}

\begin{lemma}
    $\Sperner$ is \PPP-hard.
\end{lemma}
\begin{proof}
    Same proof as for \cref{lemma:weaksperner_hard}, by reduction from \ErdosKoRado. Observe that if we have a solution of type ii) for \Sperner, the corresponding set in the \ErdosKoRado instance is either $[n]$ or $\overline{[n]}$, which is one of the desired solutions to \ErdosKoRado.
\end{proof}

\begin{lemma}
    $\Sperner \in \PPP$.
\end{lemma}

\begin{proof}
Informally, this proof is the same as the proof of \cref{lemma:weaksperner_hard}, with some additional technical details. First, we need to take care of preimages of 0. The indices corresponding to preimages of 0 correspond to solutions of type $ii)$. Second, since we only care about the first $\binom{2n}{n}$ inputs, we have to make sure that the last ones are not part of a collision, or result in a preimage of 0.

Formally, let $C: \{0, 1\}^{\alpha} \rightarrow \{0, 1\}^{2n}$ be an instance of $\Sperner$. We proceed to construct an instance of $\Pigeon$ as follows: if $x \in \{0, 1\}^{\alpha}$, we have $X:=C(x) \in \{0, 1\}^{2n}$ which is a subset of $[2n]$. Let $(X', k) = \Ecat(X)$ be the Catalan factorization of $X$, $l$ be the number of $z$'s in $X'$ and $m$ the number of bits underlined during the construction of $X'$. Note that every time we underline bits we underline simultaneously a 0 and a 1, thus $m$ is even. Then, $l = 2n - m$ is an even number. Now, let 
$$
    S(x) = \Dcat^{(l/2)}(X')
$$
Then, since $X'$ has the same number of 1's and 0's and since we replaced half of the $z$'s by 1's and the other half by 0's, we have that $S(x)$ represents an $n$-subset of $[2n]$. Informally, it is the $n$-subset of the chain that contains $X$, and replacing $z$'s by 1's enables us to move inside that chain. Finally, we set,
$$
    C'(x) = 
    \begin{cases} 
        \Ecov(S(x)) & \text{if $x < \binom{2n}{n}$}\\ 
        x & \text{if $x \geq \binom{2n}{n}$} 
    \end{cases}
$$
Then $C' : \{0, 1\}^{\alpha} \rightarrow \{0, 1\}^{\alpha}$ is an instance of $\Pigeon$ and has polynomial size. Suppose that we have a solution to this instance of $\Pigeon$ of the form $x$ such that $C'(x) = 0^{\alpha}$. Then, $x < \binom{2n}{n}$ and $\Ecov(S(x)) = 0^\alpha$ so $S(x) = \overline{[n]}$. Let $(X', k) = \Ecat(X) = \Ecat(C(x))$ be the Catalan factorization of $X$. Like previously, we get that the Catalan string that corresponds to $S(x)$ is $X'$. However, $S(x) = \overline{[n]}$ and the Catalan string that corresponds to $\overline{[n]}$ is $0^n \cat 1^n$. Thus, $X' = 0^n \cat 1^n$. Now, $C(x) = \Dcat \circ \Ecat(C(x)) = \Dcat(0^n \cat 1^n, k) = 0^n \cat 1^n$, so $C(x) = \overline{[n]}$.

Suppose instead that we have a solution to this instance of $\Pigeon$, of the form $x \neq y$ such that $C'(x) = C'(y)$. Like before, we have $x, y < \binom{2n}{n}$. Then, by injectivity of $\Ecov$ on the $n$-subsets of $[2n]$ (see \cref{lemma:cover_bij}), we get that $S(x) = S(y)$. Informally, this means that $C(x)$ and $C(y)$ belong to the same chain and thus that one is contained is the other. Let $(X', k) = \Ecat(X) = \Ecat(C(x))$ be the Catalan factorization of $X$ and $l$ be the number of $z$'s in $X'$, and let $(Y', k') = \Ecat(Y) = \Ecat(C(y))$. We have $S(x) = \Dcat(X', l/2)$ so by \cref{lemma:catalan_id2}, the Catalan string that corresponds to $S(x)$ is $X'$. Similarly, the Catalan string that corresponds to $S(y)$ is $Y'$. Since $S(x) = S(y)$, we get $X' = Y'$. We have that $X = \Dcat(\Ecat(X))$ and that $Y = \Dcat(\Ecat(Y))$ by \cref{lemma:cat_bij}, so $X = \Dcat(X', k)$ and $Y = \Dcat(Y', k') = \Dcat(X', k')$. By symmetry of $x$ and $y$ we can assume that $k \leq k'$. Then, to go from $X'$ to $X$ we added $k$ elements (the ones corresponding to the last $k$ $z$'s in $X'$) while to go from $X'$ to $Y$ we added these same $k$ elements plus $k'-k$ others. Hence, $C(x) = X \subseteq Y = C(y)$. 
\end{proof}

\begin{remark}
Like previously, the idea behind that proof is to compose our instance of $\Sperner$ with the property-preserving encoding we defined in \cref{rmk:Sperner_encoding}. However, this time it is not only the collisions that are of interest to us, but also the preimages of the 0 string.
\end{remark}

\section{Cayley's Tree Formula}
\label{sec:Cayley}
We consider yet another classic theorem from combinatorics, related to spanning trees.  A classic result by Cayley establishes the number of spanning trees of the complete graph on $n$ vertices. 
We observe then that if we have a collection of sufficiently many such graphs, either one of the graphs is not a spanning tree, or two spanning trees collide. Note that two isomorphic trees on distinct vertices are not considered a collision. This allows us to define a total search problem of either finding a collision or finding an index not corresponding to a spanning tree.
We represent trees using a bitmap on all possible edges, ordered arbitrarily.
We show that this problem is equivalent to \WeakPigeon, in a more direct way than for the previous results. As before, the problem can be modified using the same technique as previously to become equivalent to \Pigeon, and thus \PPP-complete.

\begin{classicaltheorem}[Cayley \cite{cayley}]\label{cthm:cayley}
    There are exactly $n^{n-2}$ spanning trees of the complete graph on $n$ vertices.
\end{classicaltheorem}

\begin{definition}[\WeakCayley]\label{def:problem_WeakCayley}
The problem $\WeakCayley$ is defined by the relation
\begin{description}
    \item[Instance:] A Boolean circuit $C\colon\{0, 1\}^{\lceil(n-2)\log(n)\rceil+1} \to \{0, 1\}^{\binom{n}{2}}$.
    \item[Solution:] One of the following:
    \begin{enumerate}[label=\roman*)]
    \item $x$ s.t. $C(x)$ is not a spanning tree (i.e., is not spanning, not connected or contains a cycle),
    \item $x \neq y$ s.t. $C(x) = C(y)$.
    \end{enumerate}
\end{description}
\end{definition}

\begin{theorem}\label{thm:cayley_complete}
    $\WeakCayley$ is \PWPP-complete.
\end{theorem}

For the rest of this section, we set $\beta = \lceil(n-2)\log(n)\rceil$.

\begin{lemma}\label{lemma:weakcayley_pwpp}
    $\WeakCayley \in \PWPP$.
\end{lemma}
\begin{proof}
We reduce to \WeakPigeon. Unlike the previous problems, here, we are interested in a very simple algebraic structure, namely equality. Thus, we want collisions in our encoding to correspond to equality. This means that we want an efficiently computable injective encoding of spanning trees. For this, we use Prüfer codes (\cref{sec:prufer}). We map any input $x$ to the Prüfer encoding of $C(x)$ and, therefore, a collision  either yield a collision in the trees or a graph that is not a spanning tree.

Formally, suppose that we have $C\colon\{0, 1\}^{\lceil(n-2)\log(n)\rceil+1} \to \{0, 1\}^{\binom{n}{2}}$ an instance of $\Cayley$. We may define an instance of $\WeakPigeon$ by setting $C'(x) = \Epruft(C(x))$. We observe that $C' : \{0, 1\}^{\beta + 1} \rightarrow \{0, 1\}^\beta$ is indeed an instance of $\WeakPigeon$. By definition, $C'(x)$ is the rank in the lexicographic order of the Prüfer code of $C(x)$. Now, suppose that we have a solution to this instance, that is $x \neq y \in \{0, 1\}^{\beta + 1}$ such that $C'(x) = C'(y)$. Then, $\Epruft(C(x)) = \Epruft(C(y))$. If $C(x)$ or $C(y)$ is not a spanning tree, then we have a solution to our original instance of $\Cayley$. Otherwise, $C(x)$ and $C(y)$ are spanning trees, so by injectivity of $\Epruft$ on the set of labelled spanning trees on $n$ vertices (see \cref{lemma:epruft_dpruft}), we have $C(x) = C(y)$ which is a solution to our original instance of $\WeakCayley$.
\end{proof}

\begin{remark}
Here, we can interpret $\Epruft$ as a property-preserving encoding on the set of labelled spanning trees on $n$ vertices, where the equivalence relation is equality. Hence, this is another proof of inclusion using property-preserving encodings, where we compose the instance of our problem with an appropriate property-preserving encoding. The equivalence relation has to be equality since the only spanning trees that are solutions of $\WeakCayley$ are spanning trees that are equal.
\end{remark}

\begin{lemma} \label{lemma:weakcayley_hard}
    $\WeakCayley$ is \PWPP-hard.
\end{lemma}
\begin{proof}
We interpret the output of the $\WeakPigeon$ instance as an index into the collection of all labelled spanning trees on $n$ vertices. By correctness of the encoding, the output necessarily is a spanning tree and, hence, the only solutions are collisions. We also detail some technical work to get a circuit with the right input size and output size, for which finding collisions allows solving the original instance of $\WeakPigeon$.

Formally, let $C' : \{0, 1\}^{m+1} \rightarrow \{0, 1\}^m$ be an instance of $\WeakPigeon$. 
We define a circuit $A : \{0, 1\}^{m+2} \rightarrow \{0, 1\}^m$ as follows. For any $x \in \{0, 1\}^{m+2}$, write $x = y \cat z$ with $y \in \{0, 1\}^{m+1}$ and $z \in \{0, 1\}$. Then, we set $A(x) = C'(C'(y) \cat z)$. Note that $A$ still has polynomial size and that any collision in $A$ allows us to retrieve a collision for $C'$ (like in the Merkle-Damg\aa rd construction, see \cite{merkle}). 

Let $n$ be the smallest integer such that $m + 1 \leq (n-2)\log(n)$. Note that $n$ is polynomial in $m$. Let $\beta = \lceil(n-2)\log(n)\rceil$. Then, $m + 1 \leq \beta$, hence $m + 2 \leq \beta + 1$. Now, we define a circuit $A' : \{0, 1\}^{\beta + 1} \rightarrow \{0, 1\}^{\beta - 1}$ as follows. For any $x \in \{0, 1\}^{\beta + 1}$, write $x = y \cat z$ with $y \in \{0, 1\}^{m+2}$ and $z \in \{0, 1\}^{\beta + 1 - m - 2}$. Then, we set $A'(x) = A(y) \cat z$. Note that $A'$ also has polynomial size and that any collision in $A'$ allows us to retrieve a collision for $A$ hence for $C'$.
    
Recall that we have $\Epruft : \{0, 1\}^{\binom{n}{2}} \rightarrow \{0, 1\}^{\beta}$ and $\Dpruft : \{0, 1\}^{\beta} \rightarrow \{0, 1\}^{\binom{n}{2}}$. We now define an instance $C$ of $\Cayley$ by setting $C(x) = \Dpruft(0 \cat A'(x))$. Now, suppose that we have a solution to this instance of $\Cayley$. For every $x$, $0 \cat A'(x)$ is one of the first $n^{n-2}$ elements of $\{0, 1\}^{\beta}$ in the lexicographic order, so $\Dpruft$ is well-defined and correct (i.e., it indeed returns a spanning tree) on input $0 \cat A'(x)$. Then, this solution must be $x \neq y$ such that $C(x) = C(y)$. By injectivity of $\Dpruft$ on its first $n^{n-2}$ inputs (\cref{lemma:epruft_dpruft}), we get that $A'(x) = A'(y)$ and from this we can retrieve a solution to our original instance of $\WeakPigeon$.
\end{proof}

\paragraph{\PPP-completeness using the tight bound} Again, we observe that \cref{cthm:cayley} gives an exact bound, namely that there are exactly $n^{n-2}$ labelled spanning trees on $n$ vertices. As before, this leads us to defining the following problem.

\begin{definition}[\Cayley]\label{def:problem_Cayley}
The problem $\Cayley$ is defined by the relation
\begin{description}
    \item[Instance:] A Boolean circuit $C\colon\{0, 1\}^{\lceil(n-2)\log(n)\rceil} \to \{0, 1\}^{\binom{n}{2}}$.
    \item[Solution:] One of the following:
    \begin{enumerate}[label=\roman*)]
    \item $x$ s.t. $C(x)$ is not a spanning tree and $x < n^{n-2}$,
    \item $x \neq y$ s.t. $C(x) = C(y)$ and $x < n^{n-2}$,
    \item $x$ s.t. $C(x) = T_1$ and $x < n^{n-2}$, with $T_1$ defined as in \cref{rmk:prufer_0's}.
    \end{enumerate}
\end{description}
\end{definition}

\begin{theorem}\label{thm:cayley'_complete}
    $\Cayley$ is \PPP-complete.
\end{theorem}

\begin{lemma}
    $\Cayley$ is \PPP-hard.
\end{lemma}
\begin{proof}
This proof is in spirit similar to the proof of \cref{lemma:weakcayley_hard}. We interpret the outputs of the instance of $\Pigeon$ as indices in the list of all spanning trees of the complete graph on $n$ vertices. Like in previous proofs, we have to define a circuit $A$ with sufficiently many inputs such that from any collision (resp. preimage of 0) in $A$ we can find a collision (resp. preimage of 0) in the instance of $\Pigeon$. In the instance of $\Cayley$ we create, preimages of $T_1$ correspond to preimages of 0.

Let $C' : \{0, 1\}^{m} \rightarrow \{0, 1\}^m$ be an instance of $\Pigeon$, and let $n$ be the smallest integer such that $m \leq (n-2)\log(n)$. Note that $n$ is polynomial in $m$. Let $\beta = \lceil(n-2)\log(n)\rceil$. We define $A : \{0, 1\}^{\beta} \rightarrow \{0, 1\}^\beta$ as follows. 
    $$
        A(x) = \begin{cases} 
                    C'(x) & \text{if $x < 2^m$}\\ 
                    x & \text{if $x \geq 2^m$}
               \end{cases}
    $$
If necessary, we pad the outputs of $A$ on the left by $0$'s so that they have length $\beta$ (this might be necessary for $x < 2^m$). Note that $A([2^m-1]) \subseteq [2^m-1]$ and $A$ acts as the identity over $[2^\beta-1] \setminus [2^m-1]$, hence any solution to $A$ as an instance of $\Pigeon$ immediately gives a solution to $C'$.     Recall that we have $\Epruft : \{0, 1\}^{\binom{n}{2}} \rightarrow \{0, 1\}^{\beta}$ and $\Dpruft : \{0, 1\}^{\beta} \rightarrow \{0, 1\}^{\binom{n}{2}}$. Then, we define an instance $C$ of $\Cayley$ by setting $C(x) = \Dpruft(A(x))$. 
    
Now, suppose that we have a solution to this instance of $\Cayley$. Every solution must consist of inputs $< n^{n-2}$ but $A([n^{n-2}-1]) \subseteq [n^{n-2}-1]$ by construction of $A$, and $\Dpruft$ is well-defined, correct and injective on this set by \cref{lemma:epruft_dpruft}. This implies that this solution can not be $x$ such that $C(x)$ is not a spanning tree. Then, suppose that this solution is $x \neq y$ such that $C(x) = C(y)$. By injectivity of $\Dpruft$ on $[n^{n-2}-1]$, we get that $A(x) = A(y)$ and from this we can retrieve a solution to our original instance of $\Pigeon$. Now, if this solution is $x$ such that $C(x) = T_1$ then this means that $A(x) = 0^{\beta}$ by \cref{rmk:prufer_0's} and injectivity of $\Dpruf$ over $[n^{n-2}-1]$ so $C'(x) = 0^m$.
\end{proof}

\begin{lemma}
    $\Cayley \in \PPP$.
\end{lemma}
\begin{proof}
    The idea behind the proof is similar to that of \cref{lemma:weakcayley_pwpp}, using $\Epruft$ to create an instance of \Pigeon except that we restrict the circuit to only apply the first $n^{n-2}$ elements of the collection, and set it to the identity on the rest of the inputs. Any preimage of 0 correspond to a preimage of $T_1$, and collisions arise from graphs that are not spanning trees, as well as collisions in the $\Cayley$ instance.
\end{proof}

\section{Ward-Szabo Theorem on Swell Colorings}
\label{sec:WardSzabo}
We now focus on a different theorem from extremal combinatorics, and more precisely from extremal graph theory. 
Let $G = (V, E)$ be the complete graph on $N$ vertices.
An edge-coloring $c:E \rightarrow [r]$ for some $r$ is called a \emph{swell coloring} of $G$ if it uses at least 2 colors and if every triangle is either monochromatic or trichromatic.
It is rather straightforward to see that in any $2$-coloring of $G$, there must exist a bichromatic triangle.
On the contrary, if we color each edge with a different color, we trivially get a swell coloring. The natural question that appears is then to determine the minimal number of colors required to swell-color the complete graph on $N$ vertices. This was solved in some cases by Ward and Szabo in 1995.
\begin{classicaltheorem}[Ward-Szabo \cite{swell-coloring}] \label{cthm:ward-szabo}
    The complete graph on $N$ vertices cannot be swell-colored with fewer than $\sqrt{N} + 1$ colors, and this bound is tight.
\end{classicaltheorem}

From that theorem, we can define a $\TFNP$ problem as follows: the input is a coloring $C$ of the edges of the complete graph on $2^{2n}$ vertices with $2^n$ colors, as well as three vertices $a, b, c$ such that $C(a, b) \neq C(a, c)$ to guarantee that at least 2 colors are used in the coloring. A solution is then the vertices of a bichromatic triangle (which is guaranteed to exist by \cref{cthm:ward-szabo}). We also allow extra solutions, one to specify that the edges $(a, b)$ and $(a, c)$ have the same color, and one if the coloring of the graph is not consistent.

\begin{definition}[$\WardSzabo$]\label{def:problem_WS1}
The problem $\WardSzabo$ is defined by the relation
\begin{description}
    \item[Instance:] The following:
    \begin{enumerate}
        \item A Boolean circuit $C\colon\{0,1\}^{2n} \times \{0, 1\}^{2n}  \to \{0,1\}^{n}$; and,
        \item Distinct $a, b, c \in \{0,1\}^{2n}$.
    \end{enumerate}
    \item[Solution:] One of the following:
    \begin{enumerate}[label=\roman*)]
    \item $0$ if $C(a, b) = C(a, c)$,
    \item $x, y$ s.t. $C(x, y) \neq C(y, x)$,
    \item Distinct $x, y, z$ s.t. $C(x, y) = C(y, z) \neq C(x, z)$.
    \end{enumerate}
\end{description}
\end{definition}
    
We also define two variants of this problem, whose totality is a consequence of the totality of $\SwellColoring$. \\
In the first one, we allow an extra type of solution, namely the vertices of two distinct triangles with the same ``color profile".

\begin{definition}[\textsc{Ward-Szabo-Collisions}]\label{def:problem_WS3}
The problem \textsc{Ward-Szabo-Collisions} is defined by the relation
\begin{description}
    \item[Instance:] The following:
    \begin{enumerate}
        \item A Boolean circuit $C\colon\{0,1\}^{2n} \times \{0, 1\}^{2n}  \to \{0,1\}^{n}$; and,
        \item Distinct $a, b, c \in \{0,1\}^{2n}$.
    \end{enumerate}
    \item[Solution:] One of the following:
    \begin{enumerate}[label=\roman*)]
    \item $0$ if $C(a, b) = C(a, c)$,
    \item $x, y$ s.t. $C(x, y) \neq C(y, x)$,
    \item Distinct $x, y, z$ s.t. $C(x, y) = C(y, z) \neq C(x, z)$,
    \item Two triples, $(x, y, z), (x', y', z')$, each with 3 distinct elements, s.t. $\{x, y, z\} \neq \{x', y', z'\}$ and $C(x, y) = C(x', y')$, $C(x, z) = C(x', z')$, $C(y, z) = C(y', z')$.
    \end{enumerate}
\end{description}
\end{definition}

In the second variant, we allow the same extra type of solution, namely the vertices of two distinct triangles with the same ``color profile", with the additional constraint that these triangles should be trichromatic.
    
\begin{definition}[\textsc{Ward-Szabo-Colorful-Collisions}]\label{def:problem_WS2}
The problem \textsc{Ward-Szabo-Colorful-Collisions} is defined by the relation
\begin{description}
    \item[Instance:] The following:
    \begin{enumerate}
        \item A Boolean circuit $C\colon\{0,1\}^{2n} \times \{0, 1\}^{2n}  \to \{0,1\}^{n}$; and,
        \item Distinct $a, b, c \in \{0,1\}^{2n}$.
    \end{enumerate}
    \item[Solution:] One of the following:
    \begin{enumerate}[label=\roman*)]
    \item $0$ if $C(a, b) = C(a, c)$,
    \item $x, y$ s.t. $C(x, y) \neq C(y, x)$,
    \item Distinct $x, y, z$ s.t. $C(x, y) = C(y, z) \neq C(x, z)$,
    \item Two triples $(x, y, z), (x', y', z')$, each with 3 distinct elements, s.t. $\{x, y, z\} \neq \{x', y', z'\}$, $C(x, y) = C(x', y')$, $C(x, z) = C(x', z')$, $C(y, z) = C(y', z')$ and the triangle $(x, y, z)$ is trichromatic.
    \end{enumerate}
\end{description}
\end{definition}    

\begin{theorem}\label{thm:swell_hard}
    $\WeakPigeon \leq \textsc{Ward-Szabo-Collisions} \leq \textsc{Ward-Szabo-Colorful-Collisions} \leq \SwellColoring$.
\end{theorem}

\begin{proof}
At a high level, we use the $\WeakPigeon$ circuit as the coloring of the graph. If we find a bichromatic triangle, we have found a collision. If we find two triangles with the same ``color-profile'', we have also found a collision.

Formally, let us prove that $\WeakPigeon$ reduces to \textsc{Ward-Szabo-Collisions}. Let $C : \{0, 1\}^{n+1} \rightarrow \{0, 1\}^n$ be an instance of $\WeakPigeon$. By the Merkle-Damg\aa rd construction, we can build a circuit $A : \{0, 1\}^{4n} \rightarrow \{0, 1\}^n$ of polynomial size such that finding a collision for $A$ allows finding a collision for $C$. We set $a = 0^{2n}, b = 1^{2n}$ and $c = 0^{2n-1} \cat 1$. If $A(a, b) = A(a, c)$ then we have a collision for $A$. Otherwise, we have $A(a, b) \neq A(a, c)$. We define a circuit $A' : \{0, 1\}^{4n} \rightarrow \{0, 1\}^n$ as follows. $$
        A'(x, y) = 
        \begin{cases} 
            A(x, y) & \text{if $x \leq y$}\\ 
            A(y, x) & \text{if $x > y$}
            \end{cases}
$$ Then, we define an instance of \textsc{Ward-Szabo-Collisions} by saying that the coloring is $A'$ and that $A'(a, b) \neq A'(a, c)$. 

Now, suppose that we have a solution to this instance of \textsc{Ward-Szabo-Collisions}. Note that solution cannot be $x, y$ such that $A'(x, y) \neq A'(y, x)$ by definition of $A'$.
If this solution is distinct $x, y, z$ such that $A'(x, y) = A'(x, z) \neq A'(y, z)$ then $A'(x \cat y) = A'(x \cat z)$. which implies a collision for $A$ in any case.
If this solution is two triples $(x, y, z) \neq (x', y', z')$ such that $A'(x, y) = A'(x', y')$, $A'(x, z) = A'(x', z')$, $A'(y, z) = A'(y', z')$, then by symmetry of $x, y$ and $z$, and of $x', y'$ and $z'$, we can assume $x \neq x'$. If $x = y'$ and $y=x'$, then $A'(x, z) = A'(x', z') = A'(y, z')$ and $x \neq y$ so this gives us a collision for $A$. Otherwise, from $A'(x \cat y) = A'(x' \cat y')$, from which we can find a collision for $A$. \\
In all cases, we get a collision for $A$ from which we can get a collision for $C$.
\end{proof}

\begin{theorem}\label{thm:swell3_in_pwpp}
$\textsc{Ward-Szabo-Collisions} \in\PWPP$.
\end{theorem}
\begin{proof}
We describe informally the proof. There are only $2^{3n}$ different ``color profiles" possible, which is less than the number of triangles containing the vertex $0^{2n}$. Hence, if we map sufficiently many distinct triangles containing that vertex to their color profile, it defines an instance of $\WeakPigeon$, and any solution to this instance gives us a solution of type $iv)$.

Formally, let $C : \{0, 1\}^{2n} \times \{0, 1\}^{2n} \rightarrow \{0, 1\}^n$, $a, b, c \in \{0, 1\}^{2n}$be an instance of \textsc{Ward-Szabo-Collisions}.
We consider the ``color profile" of some triangles containing the vertex indexed by $0^{2n}$.
Let $C' : \{0, 1\}^{3n+1} \rightarrow \{0, 1\}^{3n}$ be the circuit defined as follows.
For every $x \in \{0, 1\}^{3n+1}$, write $x = (y \cat z)$ with $y \in \{0, 1\}^{n+3}$ and $z \in \{0, 1\}^{2n-2}$. Then, let $y' = (1^{n-2} \cat y)$ and $z' = (10 \cat z) \in \{0, 1\}^{2n}$.
Then, we set $C'(x) = (C(0^{2n}, y'), C(0^{2n}, z'), C(y', z'))$. $C'$ defines an instance of $\WeakPigeon$. Suppose now that we have a solution to this instance of $\WeakPigeon$, that is $x_1 \neq x_2$ such that $C'(x_1) = C'(x_2)$. \\
Then, define $y_1', z_1', y_2'$ and $z_2'$ as above. Since $x_1 \neq x_2$, by construction we have that $\{0^{2n}, y_1', z_1'\} \neq \{0^{2n}, y_2', z_2'\}$ and that each of these two sets has three distinct elements. Furthermore, $C'(x_1) = C'(x_2)$ implies that $C(0^{2n}, y_1') = C(0^{2n}, y_2'), C(0^{2n}, z_1') = C(0^{2n}, z_2')$ and $C(y_1', z_1') = C(y_2', z_2')$. Hence, we have a solution of type $iv)$ to \textsc{Ward-Szabo-Collisions}.
\end{proof}

\begin{remark}
The last two theorems prove that \textsc{Ward-Szabo-Collisions} is $\PWPP$-complete. However, notice that the proof of inclusion into $\PWPP$ does not use solutions of the first three types. Hence, if we call \textsc{Ward-Szabo-Collisions}' the problem similar to \textsc{Ward-Szabo-Collisions} but without the first three types of solutions, this new problem is also $\PWPP$-complete. Indeed, the proof of inclusion into $\PWPP$ would be similar, and the proof of hardness too, only with less cases to consider. Thus, it seems (at least that is how we prove it) that what makes \textsc{Ward-Szabo-Collisions} $\PWPP$-complete is only its last type of solutions. Now, one could wonder how hard this problem becomes if we slightly modify this last type of solutions to make them harder to find. This is exactly what \textsc{Ward-Szabo-Colorful-Collisions} does.
\end{remark}

\begin{theorem}\label{thm:swell2_in_ppp}
    $\textsc{Ward-Szabo-Colorful-Collisions} \in \PPP$. 
\end{theorem}

\begin{proof}
We first give an overview of the proof. It is quite similar in spirit to the previous one, but we need to work to avoid getting collisions that would give us 2 monochromatic triangles. This costs an extra bit, hence the inclusion in $\PPP$ and not in $\PWPP$. We are given three vertices $a, b, c \in \{0, 1\}^{2n}$ such that the colors $C(a, b), C(a, c)$ and $C(a, c)$ are distinct (otherwise we have an easy solution to the instance). We create an instance of $\Pigeon$ by mapping any vertex $x$ to the pair of colors $(C(x, b), C(x, c))$ if we don't have $C(x, b) = C(x, c) = C(b, c)$ which would be a monochromatic triangle, and to the color $C(x, a)$ otherwise. We need $2n$ bits to make sure that these two types of outputs don't collide. We make sure that 0 has no preimage. Then, any solution to the instance of $\Pigeon$ must be a collision. If it is a collision from the first case, we found 2 distinct non-monochromatic triangles with the same profile, hence a solution of type $iii)$ or $iv)$. If it is a collision from the second case, we found 2 non-monochromatic triangles with the same profile.

Formally, let $C : \{0, 1\}^{2n} \times \{0, 1\}^{2n} \rightarrow \{0, 1\}^n$ and $a, b, c \in \{0, 1\}^{2n}$ be an instance of \textsc{Ward-Szabo-Colorful-Collisions}. If $C(a, b) = C(a, c)$ then we have a solution to this instance of \textsc{Ward-Szabo-Colorful-Collisions}. Now, suppose $C(a, b) \neq C(a, c)$. If $C(b, c) = C(a, b)$ or $C(b, c) = C(a, c)$, then we have a solution of type $iii)$ to this instance of \textsc{Ward-Szabo-Colorful-Collisions}. Hence, we can suppose that the colors $C(a, b), C(a, c)$ and $C(b, c)$ are all distinct. Furthermore, if $C(c, b) \neq C(b, c)$, we have a solution of type $ii)$, so we also assume that $C(c, b) = C(b, c)$. We use the circuit $E_{lex} : \{0, 1\}^n \times \{0, 1\}^n \rightarrow \{0, 1\}^{2n-1}$ defined in \cref{sec:elex}, to encode 2-subsets of $\{0, 1\}^n$ using $2n-1$ bits.

We define an instance $C' : \{0, 1\}^{2n} \rightarrow \{0, 1\}^{2n}$ of $\Pigeon$ as follows. 
\begin{align*}
    C'(x) = 
    \begin{cases}
    01110^{2n-4} & \text{if $x=a$}\\
    010^{2n-2} & \text{if $x=b$}\\
    0110^{2n-3} & \text{if $x=c$}\\
    01^{n-1} \cat C(x, a) & \text{if $C(x, b) = C(x, c) = C(b, c)$}\\
    1 \cat E_{lex}(C(x, b), C(x, c)) & \text{otherwise}
    \end{cases}
\end{align*}
Now, suppose that we have a solution to this instance of $\Pigeon$. By construction of $C'$, it cannot be $x \in \{0, 1\}^{2n}$ such that $C'(x) = 0^{2n}$. Then, it must be $x \neq y \in \{0, 1\}^{2n}$ such that $C'(x) = C'(y)$. Furthermore, by design of $C'$, we have $x, y \notin \{a, b, c\}$. We consider two cases, depending on the first bit of $C'(x)$. 
\begin{enumerate}
    \item Suppose the first bit of $C'(y) = C'(x)$ is a $1$. Then, $E_{lex}(C(x, b), C(x, c)) = E_{lex}(C(y, b), C(y, c))$. If $C(x, b) = C(x, c)$, then we have that $C(x, b) = C(x, c) \neq C(b, c)$ otherwise the first bit of $C'(x)$ would be a 0. Then, the triangle $(x, b, c)$ is bichromatic so it's a solution to our instance of \textsc{Ward-Szabo-Colorful-Collisions}. Similarly, if $C(y, b) = C(y, c)$, then the triangle $(y, b, c)$ is bichromatic. Now, if $C(x, b) \neq C(x, c)$ and $C(y, b) \neq C(y, c)$, then $\{C(x, b), C(x, c)\} = \{C(y, b), C(y, c)\}$ by injectivity of $E_{lex}$ on subsets of 2 distinct elements of $\{0, 1\}^n$. Then, $\{x, b, c\} \neq \{y, b, c\}$, each has three distinct elements, and either $C(x, b) = C(y, b)$, $C(x, c) = C(y, c)$ and $C(b, c) = C(b, c)$, or $C(x, b) = C(y, c)$, $C(x, c) = C(y, b)$ and $C(b, c) = C(c, b)$. The triangle $(x, b, c)$ is not monochromatic so this gives us a solution to our instance of \textsc{Ward-Szabo-Colorful-Collisions}, either of type $iv)$ if it is trichromatic, or of type $iii)$ if it is bichromatic.
    \item Otherwise, suppose that the first bit of $C'(y) = C'(x)$ is a 0. By construction of $C'$, this means that $C(x, b) = C(x, c) = C(b, c) = C(y, c) = C(y, b)$. Furthermore, since $C'(x) = C'(y)$, we get that $C(x, a) = C(y, a)$. Then, $\{x, a, b\} \neq \{y, a, b\}$, each has three distinct elements, and $C(x, a) = C(y, a)$, $C(x, b) = C(y, b)$ and $C(a, b) = C(a, b)$. The triangle $(x, a, b)$ is not monochromatic since $C(x, b) = C(b, c) \neq C(a, b)$ so this gives us a solution to our instance of \textsc{Ward-Szabo-Colorful-Collisions}, either of type $iv)$ if it is trichromatic, or of type $iii)$ if it is bichromatic.\qedhere
\end{enumerate}
\end{proof}

\subsection{A Hierarchy of Total Search Problems between \WeakPigeon and \Pigeon?}
In the last proof, we define a reduction to $\Pigeon$ where the circuit $C'$ only has a range of $2^{2n-1} + 2^{n-1}$ elements. Indeed, we need exactly $\binom{2^n}{2} = 2^{2n-1} - 2^{n-1}$ elements to encode the pairs of colors. We also need exactly $2^n$ elements for the fourth case. However, we can map the $x$ anywhere in that case if $C(x, a) \in \{C(a, b), C(a, c), C(b, c)\}$ because such an $x$ would give us a bichromatic triangle. Hence, we need $2^n - 3$ colors for this case. We also need 3 extra elements for $a, b$ and $c$.
Hence, overall, we only need a range of $2^{2n-1} + 2^{n-1}$ elements. 
Thus, we get a reduction from \textsc{Ward-Szabo-Colorful-Collisions} to a problem that is weaker than $\Pigeon$ (but stronger than $\WeakPigeon$), which is the following : given a circuit from $2n$ bits to $2n$ bits, either find a collision, or a preimage of one of the first $2^{2n} - (2^{2n-1} + 2^{n-1})$ elements. 

More generally, we can define the problem $\GeneralPigeon_k^m$ as follows.

\begin{definition}[$\GeneralPigeon_k^m$]\label{def:problem_GeneralPigeon}
The problem $\GeneralPigeon_k^m$ is defined by the relation
\begin{description}
    \item[Instance:] A Boolean circuit $C\colon\{0,1\}^{m} \to \{0,1\}^{m}$.
    \item[Solution:] One of the following:
    \begin{enumerate}[label=\roman*)]
    \item $x \neq y \in \{0, 1\}^{m}$ s.t. $C(x) = C(y)$,
    \item $x \in \{0, 1\}^{m}$ s.t. $C(x)$ is one of the first $k$ elements of $\{0, 1\}^{m}$.
    \end{enumerate}
\end{description}
\end{definition}

Note that this problem gets harder as $k$ decreases. It is trivial for $k = 2^{m}$, equivalent to $\WeakPigeon$ for $k = 2^{m-1}$ and to $\Pigeon$ for $k = 1$. \\
This problem induces an entire family of intermediary problems between $\WeakPigeon$ and $\Pigeon$. It is not clear how many non-equivalent problems appear in that hierarchy. It is also unclear whether each $\PWPP$-hard problem that is in $\PPP$ is in fact equivalent to one of these.  

\section{Mantel's Theorem on Triangle-Free Graphs}\label{sec:mantel}
Next, we move on to another classical theorem in extremal graph theory. It answers the following question: What is the maximum number of edges in a triangle-free graph on $N$ vertices?
\begin{classicaltheorem}[Mantel \cite{mantel}]
    If $G=(V, E)$ is a triangle-free graph on $N$ vertices then $|E| \leq N^2/4$, and this bound is tight.
\end{classicaltheorem}
This gives rise to the following search problem. Suppose that we are given a collection of strictly more than $N^2/4$ distinct edges for a graph on $N$ vertices. Then, by Mantel's theorem, there must be three of these edges forming a triangle in the graph. The search problem is then to find them. We can turn this problem into a $\TFNP$ problem if we also allow evidence that two edges in the collection are in fact the same, or that an edge is in fact a loop. For practical reasons, we demand that the endpoints of every edge are given in the lexicographic order. When the edges are represented implicitly by a poly-sized circuit, we get the following problem.

\begin{definition}[\WeakMantel]\label{def:problem_WeakMantel}
The problem $\WeakMantel$ is defined by the relation
\begin{description}
    \item[Instance:] A Boolean circuit $C\colon\{0,1\}^{2n-1} \to \{0,1\}^{n} \times \{0, 1\}^n$.
    \item[Solution:] One of the following:
    \begin{enumerate}[label=\roman*)]
    \item Distinct $i, j, k$ s.t. $C(i), C(j), C(k)$ form a triangle,
    \item $i$ s.t. $C(i) = (u, v)$ with $u \geq v$ in the lexicographic order,
    \item $i \neq j$ s.t. $C(i) = C(j)$. 
    \end{enumerate}
\end{description}
\end{definition}

\begin{remark}
Like in the other problems, the size of the collection we receive (in this case, edges) is twice the threshold size (here, $2^{n-2}$). However, here, we observe that the number of edges we receive as input is greater than the number of possible edges since $2^{n-1} > \binom{2^n}{2}$. Thus, in any instance of $\WeakMantel$, there \emph{must} be solutions of type $ii)$ or $iii)$.
\end{remark}

\begin{theorem}\label{thm:weakmantel_hard}
    $\WeakMantel$ is $\PWPP$-hard.
\end{theorem}

\begin{proof}
To prove this result, we apply the graph-hash product to the complete balanced bipartite graph on $2^n$ vertices.

Formally, let $C : \{0, 1\}^{n} \rightarrow \{0, 1\}^{n-1}$ be an instance of $\WeakPigeon$. We define $C' : \{0, 1\}^{2n-1} \rightarrow \{0, 1\}^{2n-2}$ as follows. For every $x \in \{0, 1\}^{2n-1}$, write $x = y \cat z$ with $y \in \{0, 1\}^n$ and $z \in \{0, 1\}^{n-1}$. We then set $C'(x) = C(y) \cat z$. Note that from any collision for $C'$ we can retrieve a collision for $C$ (by looking at the first $n$ bits). Now, we define $C'' : \{0, 1\}^{2n-1} \rightarrow \{0, 1\}^n \times \{0, 1\}^n$ as follows. For every $x \in \{0, 1\}^{2n-1}$, write $C'(x) = (y \cat z)$ with $y, z \in \{0, 1\}^{n-1}$. We then set $C''(x) = (0 \cat y, 1 \cat z)$. We observe that $C''$ defines an instance of $\Mantel$. Note that the edges given by $C''$ correspond to edges of the complete balanced bipartite graph on $2^n$ vertices where one side of the bipartition consists of the $2^{n-1}$ first elements in the lexicographic order. In particular, the graph described by $C''$ is triangle-free, so there is no solution of type $i)$. Similarly, by construction of $C''$, there can be no solution of type $ii)$. Thus, any solution to this instance of $\WeakMantel$ is $i \neq j$ such that $C''(i) = C''(j)$. By construction of $C''$, this means that $C'(i) = C'(j)$ and from there we can find a collision for $C$.
\end{proof}

\begin{theorem}\label{thm:weakmantel_ppp}
    $\WeakMantel \in \PPP$.
\end{theorem}

\begin{proof}
We give a high-level overview of the proof. Since we have more edges than there are possible distinct edges, we encode the edges injectively, mapping only ill-defined edges to 0. This defines an instance of $\Pigeon$, where a solution can only be a collision, meaning two different indices corresponding to the same edge.

With the circuit $E_{lex} : \{0, 1\}^n \times \{0, 1\}^n \rightarrow \{0, 1\}^{2n-1}$ defined in \cref{sec:elex}, we can encode 2-subsets of $\{0, 1\}^n$ using optimally many bits, that is $\ceil{\log \binom{2^n}{2}} = 2n-1$.

Now, consider the following circuit $E : \{0, 1\}^n \times \{0, 1\}^n \rightarrow \{0, 1\}^{2n-1}$,
$$
    E(u, v) = \begin{cases}
                    0^{2n-1} & \text{if $u \geq v$}\\
                    E_{lex}(u, v) + 0^{2n-2}1 & \text{if $u < v$}
                \end{cases}
$$
where $+$ represents the addition in binary. Note that since the range of $E_{lex}$ is exactly the first $\binom{2^n}{2}$ elements of $\{0, 1\}^{2n-1}$ in the lexicographic order, if $E(u, v) = 0^{2n-1}$, it must be that $u \geq v$.

Let $C : \{0, 1\}^{2n-1} \rightarrow \{0, 1\}^n \times \{0, 1\}^n$ be an instance of $\WeakMantel$. For every $x \in \{0, 1\}^{2n-1}$, we set $C'(x) = E(C(x))$. Then, $C' : \{0, 1\}^{2n-1} \rightarrow \{0, 1\}^{2n-1}$ is an instance of $\Pigeon$. 

Now, suppose that we have a solution to this instance of $\Pigeon$. If it is $x$ such that $C'(x) = 0^{2n-1}$, then $E(C(x)) = 0^{2n-1}$ which means that $C(x) = (u, v)$ with $u \geq v$ so $x$ is a solution to our instance of $\WeakMantel$. If it is $x \neq y$ such that $C'(x) = C'(y)$. If $C'(x) = 0^{2n-1}$, by the first case we have that $x$ is a solution to the instance of $\WeakMantel$. Now, if $C'(x) \neq 0^{2n-1}$, then it means that $E(C(x)) + 0^{2n-2}1 = E(C(y)) + 0^{2n-2}1$ so $E(C(x)) = E(C(y))$. By injectivity of $E$ on well-defined inputs (that is inputs of the form $(u, v)$ with $u < v$), this means that $C(x) = C(y)$ which is a solution to our original instance of $\WeakMantel$.
\end{proof}

\begin{remark}
Similarly to the proof that $\textsc{Ward-Szabo-Collisions} \in \PPP$, we only use the last two types of solutions, which suggests that what makes this problem easier than $\Pigeon$ is only the fact that we are given more edges than there are different possible edges in a graph on $2^n$ vertices.
\end{remark}

\begin{remark}
In fact, this last proof shows that $\WeakMantel$ reduces to $\GeneralPigeon_{2^{n-1}}^{2n-1}$. 
\end{remark}

Mantel's theorem states that there is a unique triangle-free graph on $2N$ vertices that has $N^2$ edges, it is the complete bipartite graph $K_{N, N}$. Now, consider any labelling of the vertices of $K_{N, N}$. If for every label $x$, the vertices labelled $x$ and $x+1 \mod 2N$ were on the same side of the bipartition, then all the vertices would be on the same side of the bipartition, which is impossible. Hence, there must be 2 vertices labelled $x$ and $x+1 \mod 2N$ on different sides of the bipartition, and therefore there must be an edge between them. Thus, the following problem is total.

\begin{definition}[\Mantel]\label{def:problem_Mantel}
The problem $\Mantel$ is defined by the relation
\begin{description}
    \item[Instance:] A Boolean circuit $C\colon\{0,1\}^{2n-2} \to \{0,1\}^{n} \times \{0, 1\}^n$.
    \item[Solution:] One of the following:
    \begin{enumerate}[label=\roman*)]
    \item Distinct $i, j, k$ s.t. $C(i), C(j), C(k)$ form a triangle,
    \item $i$ s.t. $C(i) = (u, v)$ with $u \geq v$ in the lexicographic order,
    \item $i \neq j$ s.t. $C(i) = C(j)$,
    \item $i$ s.t. $C(i) = (u, v)$ with $v = u + 1 \mod 2^n$ when we consider $u$ and $v$ as integers.
    \end{enumerate}
\end{description}
\end{definition}

\begin{theorem}\label{thm:mantel_hard}
$\Mantel$ is $\PPP$-hard.
\end{theorem}

\begin{proof}
To prove this result, we do the graph-hash product on the complete balanced bipartite graph on $2^n$ vertices, where one side of the bipartition consists of the first $2^{n-1}$ vertices in the lexicographic order. We make sure to map 0 into the edge $(01^{n-1}, 10^{n-1})$, which is the only edge satisfying $iv)$ in that graph.

Formally, let $C : \{0, 1\}^{2n-2} \rightarrow \{0, 1\}^{2n-2}$ be an instance of $\Pigeon$.\\
We define a circuit $C' : \{0, 1\}^{2n-2} \rightarrow \{0, 1\}^n \times \{0, 1\}^n$ as follows. Let $x \in \{0, 1\}^{2n-2}$. If $C(x) = 0^{2n-2}$, we set $C'(x) = (0 \cat 1^{n-1}, 1 \cat 0^{n-1})$. If $C(x) = 1^{n-1} \cat 0^{n-1}$, we set $C'(x) = (0^n, 1 \cat 0^{n-1})$. Otherwise, if $C(x) = (u, v)$, we set $C'(x) = (0 \cat u, 1 \cat v)$. $C'$ has polynomial size and defines an instance of $\Mantel$. 

Now, suppose that we have a solution to this instance of $\Mantel$. Like in the proof of \cref{thm:mantel_hard}, this solution cannot be of type $i)$ because the graph described by $C'$ is bipartite hence triangle-free, and it cannot be of type $ii)$ neither, by construction. If this solution is of the form $i \neq j$ such that $C'(i) = C'(j)$, by construction of $C'$ it means that $C(i) = C(j)$ which is a collision for $C$. If this solution is of the form $i$ such that $C'(i) = (u, v)$ with $v = u+1 \mod 2^n$, then by definition of $C'$, it can only be that $C'(i) = (0 \cat 1^n, 1 \cat 0^n)$. By construction of $C'$, this means that $C(i) = 0^{2n-2}$ hence $x$ is a solution to the original instance of $\Pigeon$.

\end{proof}

\subsection{Generalization with Tur\'an's Theorem}
Mantel's theorem investigates the maximum number of edges in a triangle-free graph on $N$ vertices. Similarly, one could wonder about the maximum number of edges in a graph on $N$ vertices that does not contain a clique on $r$ vertices, where $r \geq 3$ is an arbitrary constant. This problem was solved by Tur\'an in 1941. 
\begin{classicaltheorem}[Tur\'an \cite{turan}]
    If $G = (V, E)$ is a graph on $N = |V|$ vertices that does not contain any $r+1$-clique, then $|E| \leq (1-\frac{1}{r})\frac{N^2}{2}$ and this bound is tight when $r$ divides $N$.
\end{classicaltheorem}

Now, suppose that we are given a list of strictly more than $(1-\frac{1}{r})\frac{N^2}{2}$ edges for a graph on $N$ vertices. Then, by Tur\'an's theorem, if all these edges are distinct, the graph must contain an $r+1$-clique. This induces a total search, namely that of finding the vertices of such a clique. If the edges are given implicitly via a Boolean circuit which on input $i$ returns the endpoints of the $i$-th edge, we get the following $\TFNP$ problem. 
\begin{definition}[\WeakTuran]\label{def:problem_WeakTuran}
The problem $\WeakTuran$ is defined by the relation
\begin{description}
    \item[Instance:] A Boolean circuit $C\colon\{0,1\}^{2n-1} \to \{0,1\}^{n} \times \{0, 1\}^n$.
    \item[Solution:] One of the following:
    \begin{enumerate}[label=\roman*)]
    \item Distinct $i_1, i_2, \ldots i_{(r+1)(r+2)/2}$ such that $C(i_1), C(i_2), \ldots C(i_{(r+1)(r+2)/2})$ are the edges of an $r+1$-clique,
    \item $i$ s.t. $C(i) = (u, v)$ with $u \geq v$ in the lexicographic order,
    \item $i \neq j$ s.t. $C(i) = C(j)$. 
    \end{enumerate}
\end{description}
\end{definition}

\begin{remark}
Note that $r$ can be any polynomial in $n$ in the previous definition and it would still define a $\TFNP$ problem.
\end{remark}

\begin{theorem}\label{thm:turan_reduction}
For every $r_1 < r_2$, there is a reduction from $\textsc{weak-Tur\'an}_{r_1}$ to $\textsc{weak-Tur\'an}_{r_2}$.
\end{theorem}

\begin{proof}
Let $C : \{0, 1\}^{2n-1} \rightarrow \{0, 1\}^n \times \{0, 1\}^n$ be an instance of $\textsc{weak-Tur\'an}_{r_1}$. Now, we interpret it as an instance of $\textsc{weak-Tur\'an}_{r_2}$. Suppose that we have a solution to this instance of $\textsc{weak-Tur\'an}_{r_2}$. \\
If we have $(r_2+1)(r_2+2)/2$ edges that form an $r_2+1$-clique, it suffices to remove some of them to get the edges of an $r_1+1$-clique. Otherwise, any solution of type $ii)$ or $iii)$ for $\textsc{weak-Tur\'an}_{r_2}$ immediately translates into a solution of the same type for $\textsc{weak-Tur\'an}_{r_1}$.
\end{proof}

\begin{theorem}\label{thm:weakturan_hard}
For every $r \geq 2$, $\WeakTuran$ is $\PWPP$-hard.
\end{theorem}

\begin{proof}
It is enough to notice that $\textsc{WeakTur\'an}_2$ is exactly $\WeakMantel$, which is $\PWPP$-hard by \cref{thm:weakmantel_hard}. Then, apply \cref{thm:turan_reduction}.
\end{proof}

\begin{theorem}\label{thm:weakturan_ppp}
For every $r > 2$, $\WeakTuran \in \PPP$.
\end{theorem}

The proof is exactly similar to the proof of $\cref{thm:weakmantel_ppp}$. In this case too, it appears that what makes the problem easier than $\Pigeon$ is that we are given too many edges. \bigbreak

Tur\'an's theorem states that there if $r$ divides $N$, there is a unique graph on $N$ vertices that does not contain any $r+1$-clique and that has the maximum number of edges. This graph is the complete $r$-partite graph, where each part has size $N/r$. Like previously, there must be 2 vertices labelled $x$ and $x+1 \mod 2N$ with an edge between them. We denote by $N$ the largest multiple of $r$ that is at most $2^n$, and set $M = (1-\frac{1}{r})\frac{N^2}{2}$. Thus, the following problem is in $\TFNP$.

\begin{definition}[\Turan]\label{def:problem_Turan}
The problem $\Turan$ is defined by the relation
\begin{description}
    \item[Instance:] The following:
    \begin{enumerate}
        \item A Boolean circuit $C\colon\{0,1\}^{2n-1} \to \{0,1\}^{n} \times \{0, 1\}^n$; and,
        \item Two integers $N$ and $M$.
    \end{enumerate}
    \item[Solution:] One of the following:
    \begin{enumerate}[label=\roman*)]
    \item $0$ if $r$ does not divide $N$, or if $N>2^n$, or if $N+r \leq 2^n$, or if $M \neq (1-\frac{1}{r})\frac{N^2}{2}$,
    \item $i$ s.t. $C(i) = (u, v)$ with $u \geq N$ or $v \geq N$, and $i < M$
    \item Distinct $i_1, i_2, \ldots i_{(r+1)(r+2)/2}$ such that $C(i_1), C(i_2), \ldots C(i_{(r+1)(r+2)/2})$ are the edges of an $r+1$-clique, and $i_j < M$ for every $j$,
    \item $i$ s.t. $C(i) = (u, v)$ with $u \geq v$ in the lexicographic order, and $i < M$,
    \item $i \neq j$ s.t. $C(i) = C(j)$, and $i, j < M$,
    \item $i$ s.t. $C(i) = (u, v)$ with $v = u + 1 \mod 2^n$ when we consider $u$ and $v$ as integers, and $i < M$.
    \end{enumerate}
\end{description}
\end{definition}
This last problem is in $\TFNP$.
However, we cannot adapt the proof of \PPP-hardness of $\Mantel$ to it in a straightforward way and, in fact, it is open whether this problem is $\PPP$-hard.

\bibliographystyle{alpha}
\bibliography{refs}

\newpage
\appendix
\section{Efficient algorithm for the explicit Ramsey problem} \label{app:algo_ramsey}

The following proof of Ramsey's theorem is folklore. Recall the statement of the theorem \begin{description}
\item[Ramsey \cite{ramsey}]
\emph{
Any edge-coloring of the complete graph on $n$ vertices with two colors contains a monochromatic clique of size at least $\frac12 \log n$. 
}
\end{description}
\begin{proof}
Let $G = (V, E)$ be the complete graph on $n$ vertices, and $c : E \rightarrow \{0, 1\}$ be a two-coloring of its edges. \\
Pick an arbitrary vertex $v_1 \in V$. \\
$v_1$ has $n-1$ adjacent edges so at least $n/2$ of them have the same color by the pigeonhole principle. \\
Let $c_1$ be that color and $V_1 = \{v \in V \setminus \{v_1\}, c(v, v_1) = c_1\}$. \\
Then, $V_1$ has at least $n/2$ elements. \bigbreak
\noindent Next, pick an arbitrary vertex $v_2 \in V_1$. \\
There are at least $n/2-1$ edges between $v_2$ and another vertex in $V_1$. Like before, at least $n/4$ of them have the same color by the pigeonhole principle. \\
Let $c_2$ be that color and $V_2 = \{v \in V_1 \setminus \{v_2\}, c(v, v_2) = c_2\}$. \bigbreak
\noindent That way, we proceed to build by induction a finite family of vertices $(v_i)$, a finite family of colors $(c_i)$ and a finite family of sets of vertices $(V_i)$ with the following properties : \\
$\bullet$ For every $i$, $V_i \subset V_{i-1}$. \\
$\bullet$ For every $i$, $V_i$ has size at least $n/2^i$. \\
$\bullet$ For every $i$, $v_{i+1} \in V_i$. \\
$\bullet$ For every $i$ and for every $u \in V_i$, we have $c(v_i, u) = c_i$. \bigbreak
\noindent In particular, note that the second point implies that we have at least $\log(n) - 1$ $V_i$'s, thus we can construct at least $\log(n)$ $v_i$'s (since we need that $V_i$ is not empty to build $v_{i+1}$). \\
This means that we define at least $\log(n) - 1$ colors $c_i$. By the pigeonhole principle, at least $\log(n)/2$ of them are the same, say color $c \in \{0, 1\}$. \\
Let $k = \log(n)/2$. \\
Pick $i_1, i_2, \ldots, i_k$ such that $c_{i_1} = c_{i_2} = \ldots = c_{i_k} = c$. \\
We claim that the subgraph whose vertices are $v_{i_1}, v_{i_2},\ldots, v_{i_k}$ is monochromatic. \\
Indeed, let $j < l \in [k]$. \\
Then, $v_{i_l} \in V_{i_l - 1} \subset V_{i_l - 2} \subset \ldots \subset V_{i_j}$, so by the fourth point, we get that $c(v_{i_j}, v_{i_l}) = c_{i_j} = c$.
\end{proof}
\noindent Now, note that this proof is constructive and yields an algorithm to find a monochromatic subgraph of size $k = \log(n)/2$ of the complete graph on $n$ vertices. \\
In this algorithm, we have $\log(n)$ iterations, and each of them can be done in time $O(n)$, so overall we get an algorithm running in $O(n\log(n))$ time.

\end{document}